\documentclass[journal]{IEEEtran}

\usepackage{amsmath,amsfonts,amssymb}
\usepackage{amsthm}
\usepackage{algorithmic}
\usepackage{algorithm}
\usepackage{mathrsfs}
\usepackage{array}
\usepackage[caption=false,font=normalsize,labelfont=rm,textfont=rm]{subfig}
\usepackage{textcomp}
\usepackage{stfloats}
\usepackage{url}
\usepackage{verbatim}
\usepackage{graphicx}
\usepackage{cite}
\usepackage{color}
\usepackage{enumerate}
\usepackage{booktabs}
\usepackage{extarrows}
\usepackage{cases}
\usepackage{indentfirst}
\usepackage{dsfont}
\newtheorem{theorem}{Theorem}
\newtheorem{lemma}{Lemma}

\begin{document}

\title{Connection Performance Modeling and Analysis of a Radiosonde Network in a Typhoon}

\author{Hanyi Liu,
        Xianbin Cao,~\IEEEmembership{Senior Member,~IEEE},
        Peng Yang,~\IEEEmembership{Member,~IEEE}, Zehui Xiong,~\IEEEmembership{Senior Member,~IEEE},\\
        Tony Q. S. Quek,~\IEEEmembership{Fellow,~IEEE}, and
        Dapeng Oliver Wu,~\IEEEmembership{Fellow,~IEEE}
\thanks{
H. Liu, X. Cao, and P. Yang are with School of Electronic and Information Engineering, Beihang University, Beijing, China. P. Yang is also with the Department of Mathematics and Theories, Peng Cheng Laboratory, Shenzhen, Guangdong, China. (e-mail: \text{\{lhy\_0312,xbcao,peng\_yang\}@buaa.edu.cn})

Z. Xiong and T. Quek are with the Information Systems Technology and Design Pillar, Singapore University of Technology and Design, Singapore 487372. (e-mail: \text{{zehui\_xiong,tonyquek}}@sutd.edu.sg)

D. Wu is with the Department of Computer Science, City University of Hong Kong, Kowloon, Hong Kong, China. (e-mail: \text{dpwu@ieee.org})
}
}

\hyphenation{op-tical net-works semi-conduc-tor}

\maketitle

\begin{abstract}
This paper is concerned with the theoretical modeling and analysis of uplink connection performance of a radiosonde network deployed in a typhoon. 
Similar to existing works, the stochastic geometry theory is leveraged to derive the expression of the uplink connection probability (CP) of a radiosonde. 
Nevertheless, existing works assume that network nodes are spherically or uniformly distributed. Different from the existing works, this paper investigates two particular motion patterns of radiosondes in a typhoon, which significantly challenges the theoretical analysis. 
According to their particular motion patterns, this paper first separately models the distributions of horizontal and vertical distances from a radiosonde to its receiver. Secondly, this paper derives the closed-form expressions of cumulative distribution function (CDF) and probability density function (PDF) of a radiosonde's three-dimensional (3D) propagation distance to its receiver. Thirdly, this paper derives the analytical expression of the uplink CP for any radiosonde in the network. Finally, extensive numerical simulations are conducted to validate the theoretical analysis, and the influence of various network design parameters is comprehensively discussed. Simulation results show that when the signal-to-interference-noise ratio (SINR) threshold is below -35 dB, and the density of radiosondes remains under 0.01/km³, the uplink CP approaches 26\%, 39\%, and 50\% in three patterns.
\end{abstract}


\section{Introduction}
\IEEEPARstart{T}{yphoon}, commonly known as a tropical cyclonic weather event, typically develops in tropical and subtropical oceanic zones, exhibiting forceful cyclonic storms with a well-defined eye and encircling convective bands. Typhoons frequently bring about powerful winds, substantial precipitation, and surging storms, making devastating impacts on coastal regions and human activities \cite{Sobel2016HumanIO, ref2, ref3}. Understanding the formation and developmental patterns of typhoons is of significant importance for scientific inquiry, meteorological prediction, and ensuring public safety. A precise prediction of typhoons is crucial for mitigating the adverse impacts caused by these weather events. Detecting typhoons serves as an essential prerequisite in this regard. 

{A number of studies utilized Internet of Things (IoT) networks to detect typhoons \cite{iot1, iot2, iot3, iot4, iot5, iot6}. However, these networks are typically deployed on the ground or the ocean rather than within the typhoon itself, making it challenging to obtain real-time meteorological data from the upper layers of the typhoon. Consequently, typhoon detection suffers from reduced accuracy and timeliness. To overcome the limitation, the deployment of a radiosonde network becomes essential.}
Deploying a radiosonde network is recognized as a pivotal method for typhoon detection \cite{ref31, ref32, ref33}.
The radiosonde network consists of a number of radiosondes (i.e., a type of IoT devices) capable of sensing and transmitting information about the inner structure of typhoons. Receivers are installed on a high-altitude platform (e.g., an airship). 
\hspace{0.1em}Primarily, through instantaneous measurement of vertical atmospheric factors such as temperature, humidity, and pressure, radiosondes provide valuable information for scientists to comprehensively analyze the inner structure and evolutionary progression of typhoons\cite{ref4}.
Overall, the data measured and collected from the radiosonde network provides a crucial foundation for improving the accuracy of typhoon detection and forecasts \cite{ref7}.

As a type of IoT network, many significant problems (e.g., network throughput and connection probability (CP) \cite{ref9,ref10}) should be investigated for the radiosonde network.
This paper focuses on the study of CP of transmission links. Researchers in this field generally utilize stochastic geometry methods to theoretically analyze the CP of transmission links. This process involves modeling and analyzing the distribution of nodes, transmission performance, interference, etc., in a two-dimensional (2D) or three-dimensional (3D) environment. Ultimately, a closed-form expression of CP is derived through rigorous mathematical derivation. 

\subsection{Related Work}
The analysis of CP in wireless communications has garnered considerable attention from the research community.\hspace{0.2em}For instance, the authors in \cite{ref71} analyzed the CP of a reconfigurable intelligent surface (RIS) aided cellular network using the stochastic geometry theory. This work derived the approximate distributions of the composite channel gains with RIS-assisted transmission. The results represented the desired signal channel and the interference channel, respectively. \hspace{0.2em}In \cite{ref72}, a stochastic geometry approach was employed to model and analyze the coverage performance of a downlink small cell network (SCN). The impact of the network statistical parameters on the coverage performance of the SCN was discussed. 
\hspace{0.1em}In \cite{ref73}, the authors investigated the coverage performance of a cellular network consisting of multiple base stations (BSs) and users with second-order macro-diversity, in particular, with non-coherent joint transmission (NCJT). 
A stochastic geometry-based analytical framework was proposed to derive the the signal-to-interference-noise ratio (SINR) coverage probability of a typical user. 
In summary, the above works concentrated on the CP analysis of wireless networks in a 2D space without considering the effect of vertical height.

A 3D space offers a more comprehensive set of information compared to a 2D space. It allows for the consideration of an object's position not only in a plane but also along the vertical axis, enabling a more accurate description of its location and distribution. 
Therefore, many researchers turned to investigating the connection performance of 3D wireless networks.
For instance, in \cite{ref74}, the authors leveraged concepts from stochastic geometry to investigate the downlink performance of a vertical heterogeneous network (VHetNet) consisting of aerial base stations (ABSs). This work derived exact and approximate analytical expressions of a user's CP and achievable rate. 
In \cite{ref75}, \hspace{0.2em}a stochastic geometry-based approach was used for the CP modeling and analysis of single- and multi-swarm unmanned aerial vehicle (UAV) networks. Simulation results revealed that there existed an optimal height and density of UAV swarms that maximized the CP. 
The authors in \cite{ref76} utilized stochastic geometry tools to analyze the packet reception ratio (PRR) of direct communications between UAVs, with a discussion on the impact of UAV deployment height. 
In \cite{ref77}, the authors investigated the radio frequency (RF) energy harvesting and connection performance of a two hop aerial base station-aided (ABS-aided) communication system. This work obtained optimal deployment altitudes of ABSs for the best RF energy harvesting and connection performance. 
In \cite{ref78}, the coverage performance of an integrated high-altitude platform (HAP) and low-altitude platform (LAP) network was analyzed based on stochastic geometry theory. The impact of network parameters such as aerial platform altitude and LAP density was discussed.
The authors in \cite{ref79} presented an analytical framework to evaluate the downlink coverage performance using a stochastic geometry tool. In this framework, the locations of UAVs and users follow a typical Poisson cluster process (PCP) distribution. 
Simulation results revealed that the height and the density of UAVs had significant influence on the network coverage probability. 
The research mentioned above delved into the analysis of network performance within a 3D environment and properly considered the influence of the height parameter. 
Nonetheless, the height parameter was regarded as a variable rather than a random variable subject to a specific distribution in the aforementioned studies. \hspace{0.2em}As a result, \hspace{0.2em}theoretical analysis methods proposed to evaluate the connection performance of 2D wireless networks can be applied to some extent.

The recent studies in \cite{ref82,ref83,ref85,ref38, ref11, ref14, ref16, ref17 ,ref21, ref22, ref24, ref27, ref331, ref341, ref35,ref12, ref15, ref39, new5} have investigated the connection performance of 3D wireless networks with a specific assumption on the statistical distribution of node deployment altitude.
For example, in \cite{ref82}, the authors proposed a novel stacked Poisson line hardcore point process (PLHCP) by introducing the randomness and the safety distance in the altitude dimension. Based on PLHCP, \hspace{0.2em}two tractable approximations were proposed to investigate the uplink CP of a network node.
The authors in \cite{ref83} studied the connection performance of a  non-orthogonal multiple access (NOMA) enhanced UAV-to-Everything (U2X) network, where receivers were assumed to be distributed in a 3D sphere. They derived the closed-form expressions for the outage probability and the ergodic rate of the paired NOMA receivers using stochastic geometry tools.
\hspace{0.2em}Besides, \hspace{0.2em}in \cite{ref85}, \hspace{0.2em}stochastic geometry tools were utilized to compare the CP and the average achievable rate of two 3D wireless aerial networks. The considered networks were a truncated octahedron-based aerial network and a binomial-Voronoi aerial network, where a 3D Poisson point process (PPP) was leveraged to model the locations of aerial nodes in both networks. {Nevertheless, the above works did not investigate the impact of 3D environment and node mobility on theoretical analysis. These factors directly affected spatial distributions of network nodes, which in turn affected the performance analysis results.}

{To this aim, some works in recent two years discussed the impact of their influences.}
{For instance, 
in \cite{new2}, the authors designed a mathematical framework to characterize the performance of an autonomous underwater vehicles (AUVs) network and theoretically analyzed the connection performance. In the AUVs network, multiple AUVs dynamically changed their 3D locations according to the random way-point (RWP) mobility model in a given underwater region.
Drawing on a stochastic geometry theory, the authors in \cite{new3} studied the downlink coverage probability in a mmWave network of ground bases stations assisted by a finite number of UAV base stations under 3D blockage effects of buildings.
Besides, in \cite{new4}, the authors proposed interference coordination via power control under 3D blockage effects in urban environments and derived a theoretical expression of coverage probability of UAV networks with Nakagami-m fading assumption based on stochastic geometry.}

\subsection{Motivations and Challenges}
The existing studies 
indicate that the theoretical analysis of connection performance in wireless communication networks in 3D space has become a significant development trend in the research field \cite{ref12, ref15, ref39}. Researchers typically utilize the stochastic geometry theory to conduct the theoretical analysis of network connection performance. 
However, most of the existing works focused on node-dense networks and assumed that network nodes were uniformly or spherically distributed in 3D space. 
As far as we know, the connection performance of a radiosonde network deployed in a typhoon has not been investigated. 

{Nevertheless, a number of researchers have recently discussed the link quality of communication systems under extreme weather conditions (e.g., typhoons) \cite{tfnew5,chu2020underwater,tfnew2,tfnew8}. For example, the authors in \cite{tfnew5} studied the evaporation duct's impact on electromagnetic (EM) wave propagation during a typhoon and the effects of a typhoon eye on EM wave propagation at different signal frequencies and antenna heights. The authors in \cite{chu2020underwater} investigated the effects of typhoons on underwater optical path loss. The results showed that typhoons significantly reduced the maximum detectable depth of underwater electro-optical identifiers by about 15 m.
Besides, the authors in \cite{tfnew2} proposed to enhance the link quality using a multi-connectivity strategy such that communication systems can adapt to extreme weather conditions. 
However, few of them investigated the impact of movement characteristics of a typhoon on network performance.}

{
The movement characteristics of a typhoon are extremely complex, and its internal aerodynamic and fluid dynamics have a strong influence on the spatial distribution of radiosondes within it.} 
{Therefore, it is essential to review the studies on the internal aerodynamic and fluid dynamics of a typhoon.} 
{Currently, some scholars are focusing on the study of the internal structure and fluid mechanics of typhoons by analyzing the development patterns of past typhoons\cite{tf2, tf1, tf3, tf4, tf5, tf6, tf7, tf8, tf9, tf10, tf11, tf12}. For instance, the authors in \cite{tf2} focused on the precipitation intensification and ice-phase microphysical processes in the spiral rainband of typhoons. The study showed that the cold rain process was crucial for the precipitation intensification in the spiral rainband, with the production and growth of graupel being the most effective process for this intensification. In \cite{tf5}, the authors proposed a mathematical model to study the development of a tropical cyclone. The growth and maintenance of a tropical cyclone mainly depended on the air updraft caused by the latent heat of steam condensation. This mechanism is accompanied by the Coriolis force, which causes the circular air motion, and by the friction of water drops (or ice pieces) against air, which decelerates the air updraft. In \cite{tf12}, the authors investigated the role of the upper-level vertical wind shear (VWS) on the rapid intensification of typhoons. The simulation showed that under moderate easterly VWS, the tilting-induced convective asymmetry moved from the downshear to upshear quadrant, wrapping around the storm center, which enhanced upward motion at the upshear flank and generated upper-level divergent flow.}

{From the research results, we may conclude that} the motion patterns of radiosondes in a typhoon are particular, and we cannot directly apply the 
theoretical analysis approaches proposed in existing works \cite{ref82, ref83, ref85, ref38, ref11, ref14, ref16, ref17 ,ref21, ref22, ref24, ref27, ref331, ref341, ref35,ref12, ref15, ref39, new5, new2, new3, new4} {to modeling and analyzing CP of a radiosonde network in a typhoon.} 
Additionally, during the process of deriving formulas such as the cumulative distribution function (CDF) and probability density function (PDF) of direct propagation distance, complex functional, differential, and integral operations are encountered. This poses a significant challenge to the theoretical derivation.

\subsection{Contributions}
This paper aims to theoretically analyze the connection performance of a radiosonde network deployed in a typhoon. The main contributions of this paper are summarized as below:
\begin{itemize}
    \item We investigate two particular motion patterns of radiosondes in a typhoon. Based on the motion patterns of radiosondes, we separately model the distributions of horizontal and vertical distances from a radiosonde in a typhoon to its receiver. With the obtained distributions, we derive the closed-form expressions of CDF and PDF of the radiosonde's 3D propagation distance to its receiver.
    \item Based on the 3D distance distribution, we further derive the analytical expression of the uplink CP for any radiosonde in the network. 
    \item Extensive numerical simulations are conducted to validate the theoretical analysis. Besides, the impact of various network design parameters is comprehensively discussed, and some useful insights for deploying a radiosonde network are revealed. Simulation results indicate that when the SINR threshold drops below -35 dB, and the density of radiosondes stays below 0.01/km³, the uplink CP reaches around 26\% and 50\% in two distinct patterns and 39\% in a hybrid pattern.
\end{itemize}

This paper is organized as follows. Section \uppercase\expandafter{\romannumeral2} presents the system model, horizontal distance models, and vertical distance models. The theoretical analysis of the uplink CP is performed in Section \uppercase\expandafter{\romannumeral3}. 
Numerical results are presented and discussed in Section \uppercase\expandafter{\romannumeral4}. This paper is concluded in Section \uppercase\expandafter{\romannumeral5}. Besides, a list of the key mathematical notations is summarized in Table \uppercase\expandafter{\romannumeral1}.
\\

\begin{table*}[!t]
  \renewcommand{\arraystretch}{1.2}
  \caption{LIST OF KEY NOTATIONS}
  \label{tab:table1}
  \newcommand{\tabincell}[2]{\begin{tabular}{@{}#1@{}}#2\end{tabular}}
  \centering
  \begin{tabular}{|l|l|l|l|}
   \hline
   {\tabincell{l}{Notation}} & {Definition} & {Notation} & {Definition}  \\\hline
  
   \tabincell{l}{$\mathcal{Z}$} & \tabincell{l}{Set of the interference radiosondes} & {$\mathcal{P}_t$} & {Transmit power of a radiosonde}  \\ \hline
   {\tabincell{l}{$R_i$, $r_i$, $i\in \{TR,\mathcal{Z}\}$}} & {Set of horizontal distance from a radiosonde to RR} &  \tabincell{l}{$\mathcal{P}_r$} & {Received power of the RR}  \\ \hline
    {\tabincell{l}{$H_i$, $h_i$, $i\in \{TR,\mathcal{Z}\}$}} & {Set of vertical distance from a radiosonde to RR} &  {$g_t$}  & {Transmitting antenna gain}  \\ \hline
    {\tabincell{l}{$L_i$, $l_i$, $i\in \{TR,\mathcal{Z}\}$}} & {Set of direct propagation distance from a radiosonde to RR} &  {$g_r$ }  & {Receiving antenna gain
}  \\ \hline
    {\tabincell{l}{$\theta_i$, $i\in \{TR,\mathcal{Z}\}$}} & {Set of elevation angle from a radiosonde to RR} &  {$A_{r_T}$}  & {Rain attenuation of TR link}  \\ \hline
   \tabincell{l}{$\mathcal{K}_i, i \in \left\{TR,\mathcal{Z}\right\}$} & {Set of all communication links} & {$\mathcal{M}$}  & {Set of rain attenuation of IR links} \\ \hline
   {$g_i, i\in \{TR,\mathcal{Z}\}$} &  \tabincell{l}{Small-scale Rayleigh fading caused by abundant \\ scatterers}  & {$\gamma_R$} & {Rain attenuation ratio} \\ \hline
   {${d_s}_i,i\in \{TR,\mathcal{Z}\}$} &  {Distance through the rainfall region}  & {$r$} & {Distance correction factor} \\ \hline
   {$\alpha$} &  {Path loss exponent}  & {$\epsilon$} & {Power control factor} \\ \hline
   {$\lambda_n$, $\lambda_r$} & {Distribution density of radiosondes} &  {$\sigma^2$}  & {Thermal noise of the RR} \\ \hline 
   {$k_s$} & {Shape parameter in Weibull distribution} &  {$\lambda_s$}  & {Scale parameter in Weibull distribution} \\ \hline  
   {$p_c$} & {Connection probability of the uplink} & {$R^{max}_{RX}$}  & \tabincell{l}{Maximum horizontal distance within the \\ communication range of RR} \\ \hline 
   {$H^{max}_{RX}$} & {Maximum vertical distance} & {$H^{min}_{RX}$}  & {Minimum vertical distance} \\ \hline  
  \end{tabular}
 \end{table*}

\section{System, Power, and Distribution Models}
\subsection{System Model}
\begin{figure}[!t]
\centering
\includegraphics[width=3in]{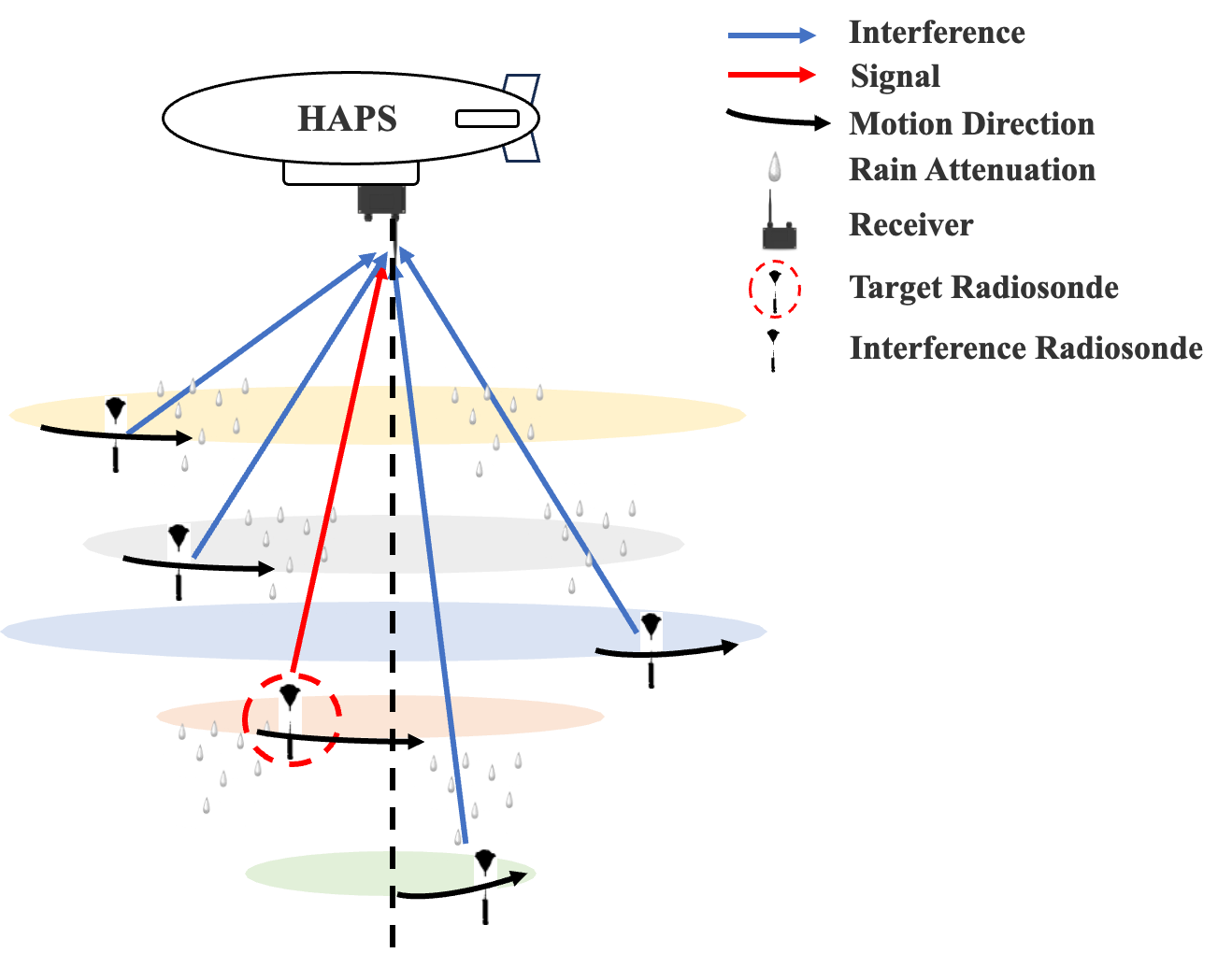}
\caption{A scenario of collecting meteorological information from a number of radiosondes.}
\label{fig:enter-label}
\vspace{-0.3cm}
\end{figure}

In this paper, we theoretically analyze the uplink communication performance of a radiosonde system containing a high altitude platform station (HAPS) with a radiosonde receiver (RR), a randomly selected target radiosonde (TR), and a number of interference radiosondes (IRs), as shown in Fig. \ref{fig:enter-label}. The HAPS is considered to be deployed above the eye of the typhoon. The location of the HAPS $\mathcal{O} \in \mathbb{R}^3$ is regarded as the origin of the coordinates. The TR will send perceived meteorological information to the RR in the presence of signal interference from IRs. The interference set is denoted by $\mathcal{Z}=\left\{IR_1, IR_2,IR_3,...,IR_N\right\}$. All communication links (TR-RR, IRs-RR) denoted by $\mathcal{K}_i, i \in \left\{TR,\mathcal{Z}\right\}$ are unidirectional and share the frequency band. All the radiosondes are equipped with omni-directional antennas and utilize the same transmit power $\mathcal{P}_t$. 
Radiosondes are dropped and distributed in different areas of the typhoon, and their trajectories are illustrated in Fig. \ref{fig:enter-label2}.
This type of movement can be approximated as circular motion. Denote the coordinate of the center point of the trajectory of a radiosonde by $\mathcal{O}^\prime_i$, $i \in \left\{TR,\mathcal{Z} \right\}$. 
The movement radius of a radiosonde is denoted by $R_i$. The distance between $\mathcal{O}^\prime_i$ and $\mathcal{O}$, denoted by $H_i$, is a vertical distance.
The direct propagation distance $L_i$ between a radiosonde and the RR as well as the elevation angle $\theta_i$ can be derived according to $H_i$ and $R_i$.

\begin{figure}[!t]
\centering
\includegraphics[width=2.5in]{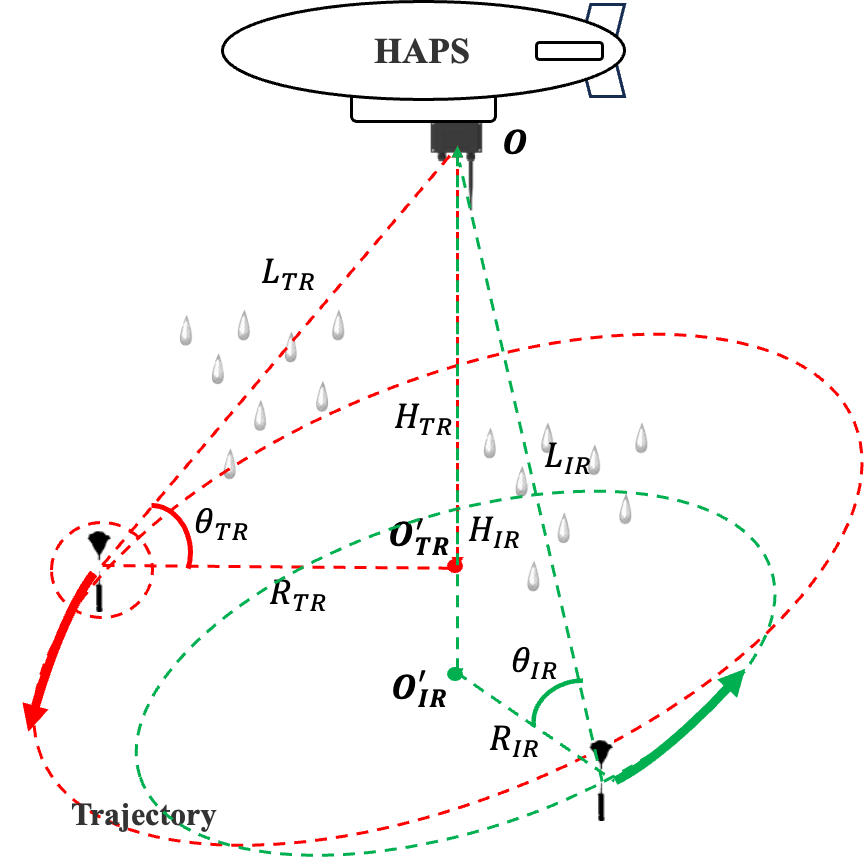}
\caption{Distances from TR and IR to RR, and trajectories of radiosondes.}
\label{fig:enter-label2}
\hspace{-0.3cm}
\end{figure}

\begin{figure}[!t]
\centering
\includegraphics[width=3.5in]{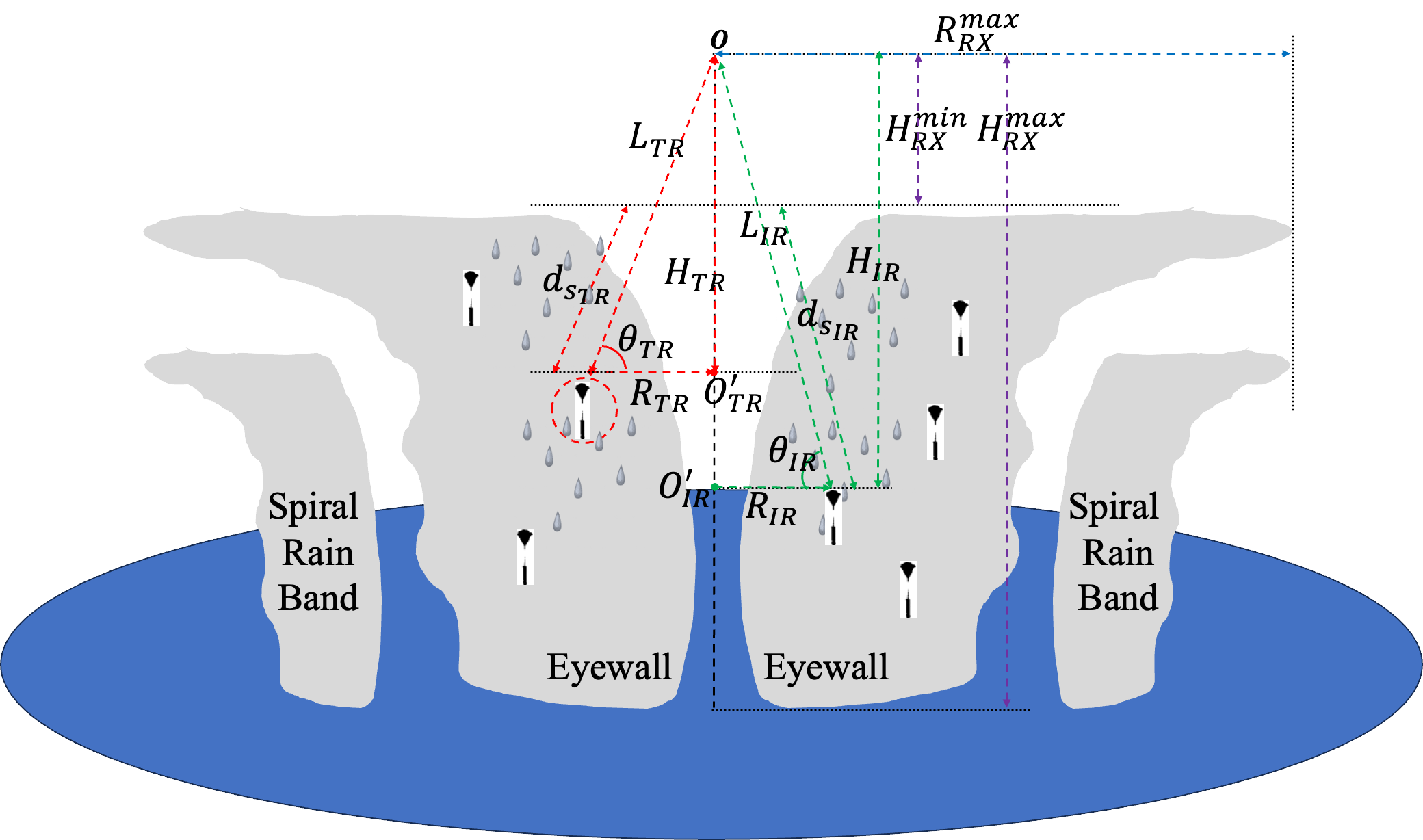}
\caption{Typhoon structure and various types of distances in the sectional view.}
\label{fig:enter-label3}
\vspace{-0.2cm}
\end{figure}

\subsection{Receiving Power Modeling}
In a typhoon, there exhibits heavy rainfall regions (exactly, the eye wall region) near the typhoon center and spiral rain bands, as shown in Fig. \ref{fig:enter-label3}. Due to the presence of heavy rain, the rain attenuation effect on the transmission of perceived meteorological information by radiosondes may be significant. 
Denote $A_{r_T}$ as the rain attenuation over the link $\mathcal{K}_{TR}$, and denote $\mathcal{M}=\left\{A_{r_\mathcal{Z}}(IR_1),A_{r_\mathcal{Z}}(IR_2),A_{r_\mathcal{Z}}(IR_3),...,A_{r_\mathcal{Z}}(IR_N)\right\}$ as the set of rain attenuation over the interference links.
As recommended by the International Telecommunication Union (ITU), for the link $\mathcal{K}_{TR}$, its rain attenuation (in dB) can be computed by \cite{Specific2023ITU}
\begin{align}
    A_{r_T} &= rd_{s_i}\gamma_R \\
    \gamma_R &= aR^b \label{eq2}
\end{align}
where $\gamma_R$ is rain attenuation ratio with the unit of $dB/km$, $r$ is distance correction factor, ${d_s}_i$ (in km) is the distance of transmission through the rain region, $a$ and $b$ are related to the carrier frequency $f$ (in GHz), polarization, and path elevation angle $\theta_i$, which can be computed by \cite{Specific2019ITU}
\begin{numcases}{}
    a = 2.5292 \times 10^{-7} f ^{5.8688 - 1.2697\log{f}}\label{eq3}
    \\
    b = 2.2698 - 1.2145 \log{f} + 0.2293\log{^2f}
    \\
    {d_s}_i = \frac{H_i - H^{min}_{RX}}{\sin{\theta_i}} 
\end{numcases}
and the minimum safe vertical distance between HAPS and typhoon is denoted by $H^{min}_{RX}$.

Besides, the abundant scatterers (such as ice crystals) in the upper layer of the typhoon lead to scattering of signals and form multipath propagation. Thus, we consider small-scale Rayleigh fading \hspace{0.2em}$g_i$\hspace{0.2em} over the link \hspace{0.2em}$\mathcal{K}_i$, which is a random variable obeying an exponential distribution with a mean of $\mu ^{-1}$. Mathematically, $g_i\sim \text{exp}{(\mu)}, i \in \left\{TR,\mathcal{Z}\right\}$. The path loss is inversely proportional to the propagation distance with a path loss exponent $\alpha$. $\epsilon \in \left[0,1\right]$ is the power control factor that controls the power of the compensation of path loss. Then, the effective signal receiving power at RR can be expressed as 
\begin{equation}
    \mathcal{P}_r = \mathcal{P}_t g_{TR} A^{-1}_{r_T} L^{-\alpha}_{TR}H^{\alpha\epsilon}_{TR}g_tg_r
\end{equation}
where \hspace{0.2em}$H^{\alpha\epsilon}_{TR}$\hspace{0.2em} is a power control item, \hspace{0.2em}$g_t$\hspace{0.2em} and \hspace{0.2em}$g_r$\hspace{0.2em} are the transmitting antenna gain and the receiving antenna gain, respectively, which remain unchanged for all the radiosondes.

\subsection{3D Spatial Distribution of Radiosondes}
Owing to the uncertainty of the radiosonde's motion, 
the spatial distribution of radiosonde cannot be modeled using conventional methods such as spherical modeling in a 3D environment. 
Nevertheless, the horizontal distance $R_i$ and vertical distance $H_i$ are independent of each other. Then, we attempt to model the spatial distribution of a radiosonde in horizontal and vertical directions separately.

\subsubsection{Horizontal Model}
In this subsection, we employ a homogeneous Poisson point process (HPPP) to model the network of radiosondes. This model 
guarantees a relatively even dispersion of radiosondes. In a HPPP, the positions of radiosondes are mutually independent.
As mentioned above, radiosondes are generally considered to follow circular motion. Nevertheless,
they might undergo irregular motion in response to wind directions. 
Thus, we will undertake diverse modeling approaches based on pattern of the motion, and the in-depth discussions are outlined below. 
\begin{enumerate}[i)]
    \item Circular Motion Model \\
    \begin{figure}[!t]
        \centering
        \includegraphics[width=3in]{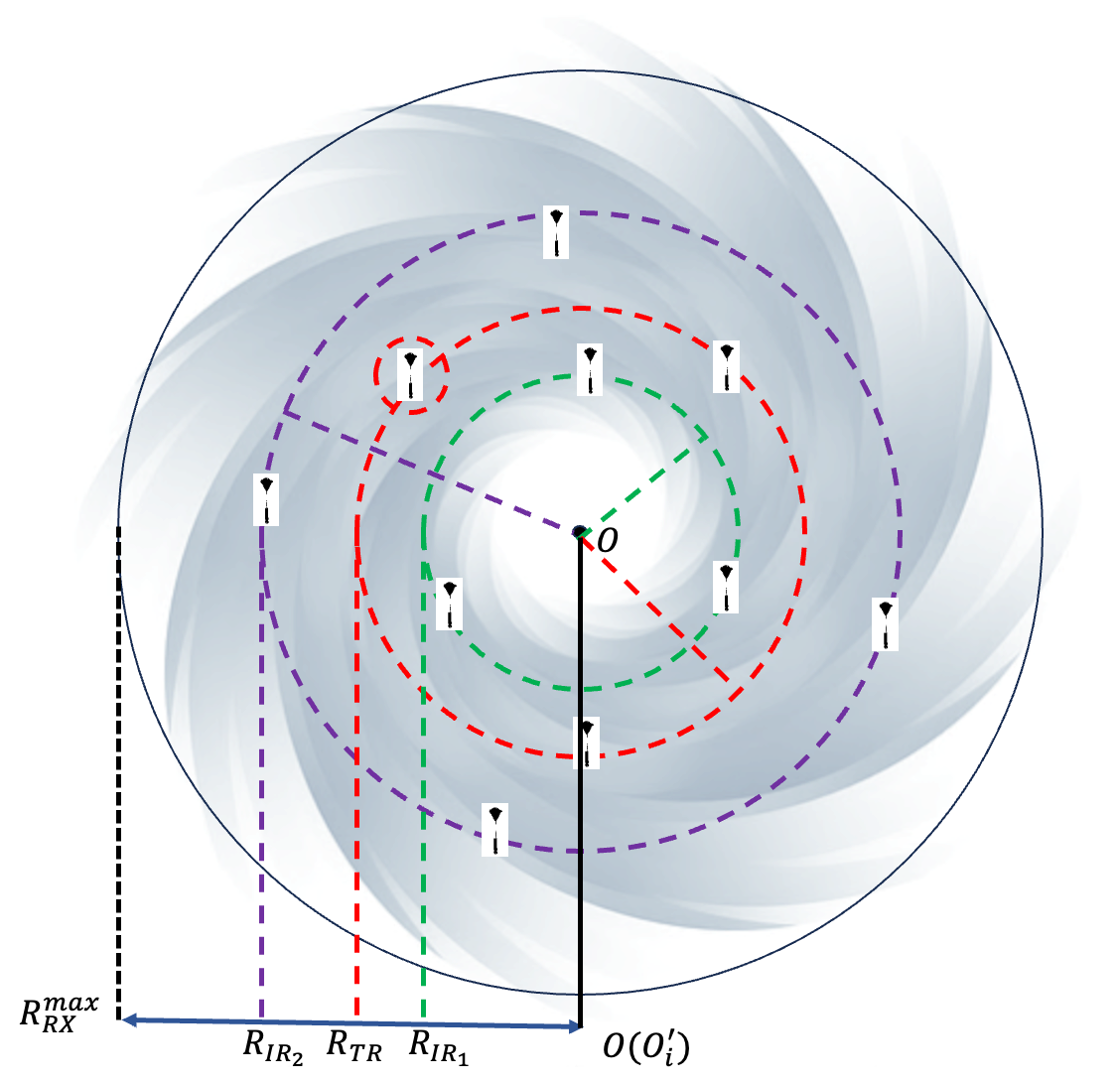}
        \caption{Circular motion of radiosondes.}
        \label{fig:enter-label4}
        \vspace{-0.3cm}
    \end{figure}
\hspace{1em}Contemplating a 2D plane, a radiosonde's trajectory forms a consistent circle with an unchanging radius \cite{tf5}. The radius of the trajectory is equal to the horizontal distance from the radiosonde to the RR. We then have the option to project the 2D-HPPP $\Phi_r$ within the plane as a one-dimensional (1D) HPPP along a line, as depicted in Fig. \ref{fig:enter-label4}. After the projection, some radiosondes may be located at the position with the same horizontal distance toward the RR. \hspace{0.2em}The peak value $R^{max}_{RX}$ on this linear segment signifies the maximum position that the radiosonde can reach in the horizontal directions within the communication range of the RR. The trajectory radii of all the radiosondes adhere to a 1D-HPPP in the horizontal direction. From the definition of 1D-HPPP, \hspace{0.2em}the number of trajectory radii $N_{R_i}$ is a random variable obeying the Poisson distribution with a mean of ${\lambda}_r R^{max}_{RX}$, where  $\lambda_r$ represents the distribution intensity of radiosondes. The radius $R_i$ is uniformly distributed in a range of $[0 , R^{max}_{RX}]$. The CDF and PDF of $R_i$ are given by

\begin{equation}\label{CDF_Ri}
    F_{R_i}(r_i) = \frac{\pi(r_i)^2}{\pi(R^{max}_{RX})^2} = (\frac{r_i}{R^{max}_{RX}})^2
\end{equation}
\begin{equation}\label{PDF_Ri}
    f_{R_i}(r_i) = \frac{dF_{R_i}(r_i)}{dr_i} = \frac{2r_i}{(R^{max}_{RX})^2}
\end{equation}

    \item Irregular Motion Model \\
    \begin{figure}[!t]
        \centering
        \includegraphics[width=2.5in]{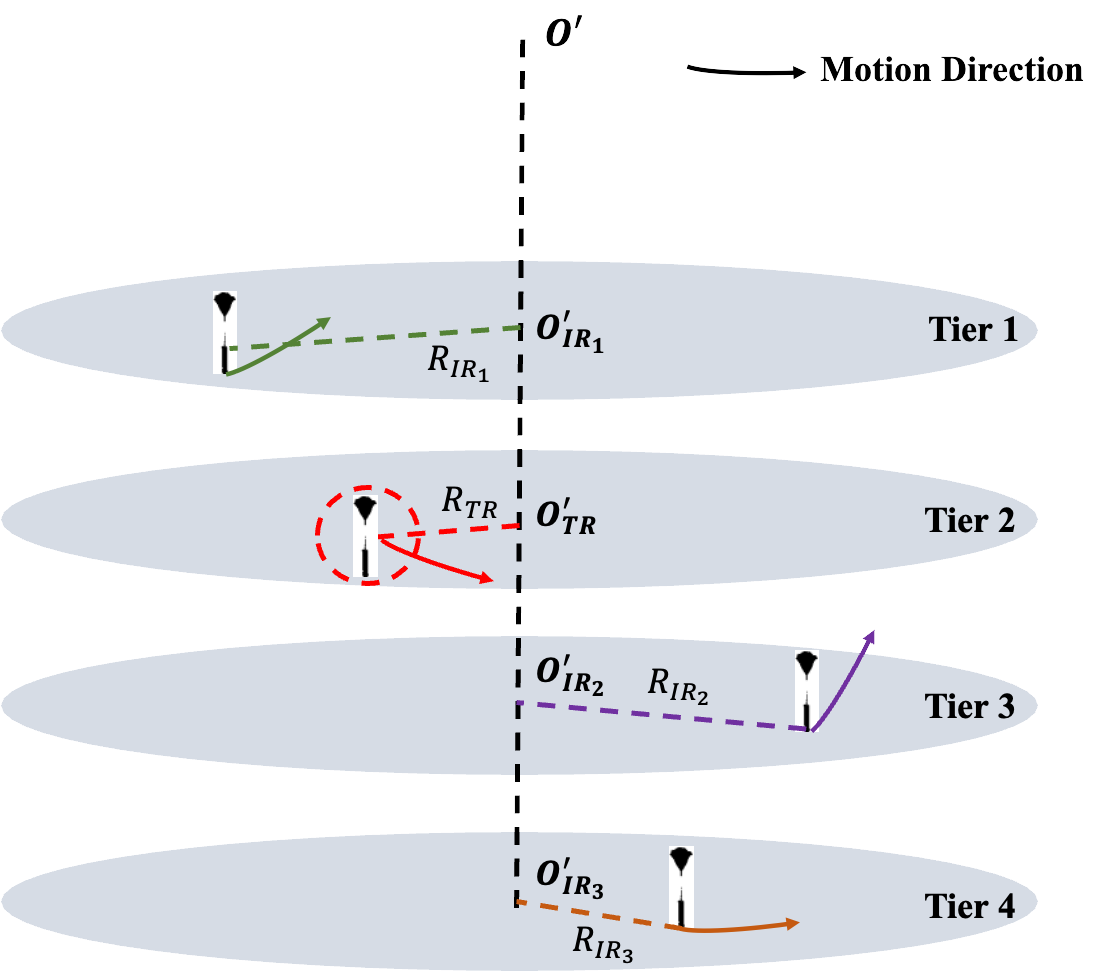}
        \caption{Radiosondes' motion directions and their distributed altitude tiers.}
        \label{fig:enter-label6}
        \vspace{-0.3cm}
    \end{figure}
    \hspace{1em}Radiosondes may exhibit an irregular motion state in a typhoon, \hspace{0.2em}making it challenging to accurately articulate the motion model with specific formulas\cite{tf12}. The precise dynamics remain elusive, given the absence of current literature exploring the motion of radiosondes. 
    Consequently, we commence the modeling investigation from the standpoint of network topology. Represent all the radiosondes in the entire space as a 2D-HPPP $\Phi_n$ with an intensity denoted by $\lambda_n$. \hspace{0.2em}Radiosondes are dispersed across different altitude tiers and move in different directions, as shown in Fig. \ref{fig:enter-label6}. For each radiosonde $i$, we can find its nearest horizontal distance to the corresponding center point $\mathcal{O}^\prime_i$, $i \in \left\{TR,\mathcal{Z} \right\}$ in each tier.
    According to the stochastic nature of HPPP, $R_i$ can be any value within the range $\left(0, +\infty\right)$. Nevertheless, the communication range of a radiosonde is finite, capped at a maximum value of $R^{max}_{RX}$. $R_i$ is thus confined to the interval $\left[0, R^{max}_{RX}\right]$, and its CDF is expressed as
    \begin{equation}\label{case1FRi}
    \begin{split}
        \mathcal{F}_{R_i}\left(r_i\right) &= \mathbb{P}\left[R_i  \leq r_i\right] \\
        &= 1 - \mathbb{P}\left[R_i > r_i\right] \\
        &= \frac{1 - \exp{\left(-\lambda_n \pi r^2_i\right)}}{{1 - e^{-\pi\lambda_n{(R^{max}_{RX})}^2}}}
    \end{split}
    \end{equation}
    where $\mathbb{P}\left[R_i > r_i\right]$ can be construed as the likelihood of no points emerging within a circular region with a radius $R_i$.
    The derivation of Eq. (\ref{case1FRi}) results in the PDF for $R_i$
    \begin{equation}
    \label{eq:irregular_horizontal_model}
        f_{R_i}\left(r_i\right) = \frac{2\pi\lambda_nr_ie^{-\pi\lambda_nr^2_i}}{{1 - e^{-\pi\lambda_n{(R^{max}_{RX})}^2}}} , r_i \in \left[0, R^{max}_{RX}\right]
    \end{equation}
\end{enumerate}

\subsubsection{Vertical Model}
For any radiosonde $i$, its altitude $H_i$ is affected by various dynamic and intricate meteorological factors, which challenges its distribution modeling.
The radiosonde network differs greatly from the existing networks, such as urban mobile networks, satellite networks, and UAV networks, where the altitude of a node is at times perceived as constant or uniformly distributed. 
Further, an insufficient exploration and validation exists in previous research concerning the vertical motion of network nodes in complex environments. 
To this end, 
considering the general direction of wind within a typhoon, we utilize the Weibull distribution to model $H_i$. The general PDF of $H_i$ is
\begin{equation}
    f_{H_i}\left(h_i;\lambda_s, k_s\right) = \frac{\frac{k_s}{\lambda_s}\left(\frac{h_i}{\lambda_s}\right)^{k_s - 1}e^{-\left(\frac{h_i}{\lambda_s}\right)^{k_s}}}{{e^{-\left(\frac{H^{min}_{RX}}{\lambda_s}\right)^{k_s}} - e^{-\left(\frac{H^{max}_{RX}}{\lambda_s}\right)^{k_s}}}}
\end{equation}
where $\lambda_s$ and $k_s$ are the scale and shape parameters of the Weibull distribution, and $h_i \in [H^{min}_{RX}, H^{max}_{RX}]$ with $H^{max}_{RX}$ representing the farthest vertical distance from a radiosonde to the RR that the radiosonde can reach in the vertical direction.

\subsection{{3D Spatial Distribution Modeling}}
In this subsection, we discuss the 3D spatial distribution modeling of radiosondes in two cases. 
\subsubsection{Case 1}
{
During the development phase of a typhoon, intense convective activity occurs. The typhoon draws moisture from the ocean surface, causing warm air to rise rapidly and form strong convection layers. Cold anomalies at the tropopause locally destabilize the atmosphere, which enhances the deep convection and thus promotes the development of outflow channels.} {The following two equations model the directions of fluids in some regions of a typhoon \cite{tf12}.}
{
\begin{equation}{\label{coolingeffect}}
    Q_{d y n}=-\bar{w}\left(\partial_z T+\frac{g}{C_p}\right)
\end{equation}
\begin{equation}{\label{velocity}}
    \frac{dw}{dt}=g\frac{T^{\prime}}{\overline{T}}
\end{equation}
where $Q_{d y n}$ is the dynamical cooling effect, $\bar{w}$ is the vertical velocity, $T$ is temperature, $C_p$ is the specific heat of dry air at constant pressure, $g$ is the gravity constant, $-\partial_z T$ represents the environmental lapse rate, and $\frac{g}{C_p}$ is the dry adiabatic lapse rate. $\bar{T}$ represents the azimuthal mean, and $T^\prime$ indicates the asymmetric component deviated from the azimuthal average, which calculates the vertical velocity tendency due to the local buoyancy production. 
}

{Eq. (\ref{coolingeffect}) shows that when updrafts lift the atmosphere, the adiabatic cooling effect of the airflow causes the temperature to drop, creating a cold anomaly. The dynamic cooling effect is particularly obvious at high altitude, which helps reduce the atmospheric static stability, thereby promoting further convection and intensification of typhoons. The buoyant acceleration is computed by Eq. (\ref{velocity}) and describes the process by which an air mass accelerates upward due to buoyancy. When the temperature of the convective region $T^\prime$ increases relative to the environmental temperature $\bar{T}$, it will lead to greater buoyant acceleration, resulting in stronger upward motion. This upward movement contributes to the enhancement of convective activity and the formation of eyewall structures.} 
\begin{figure}[!t]
\centering
\includegraphics[width=3in]{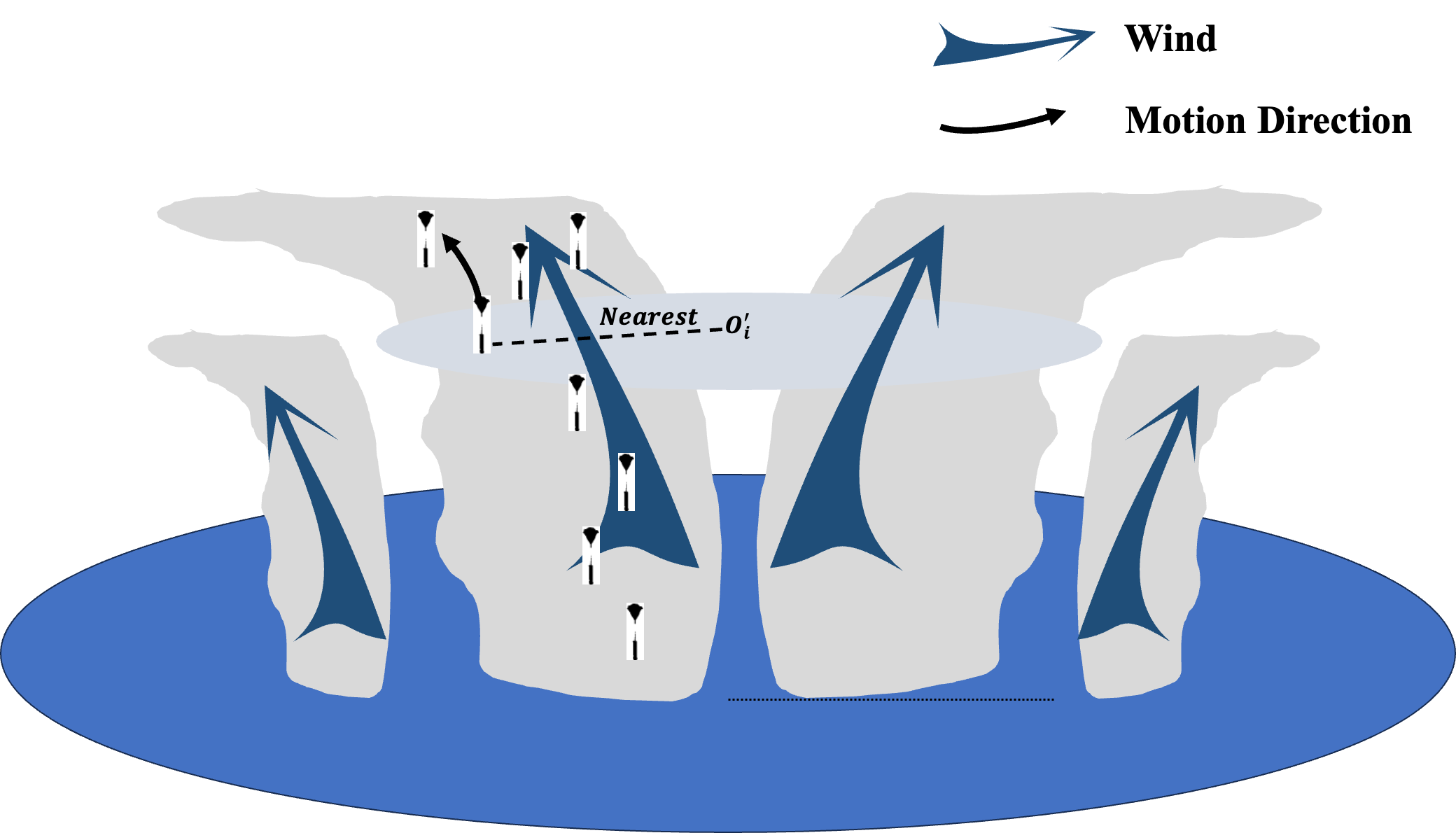}
\caption{Wind and radiosonde motion directions in Case 1.}
\label{fig:3dmodel1}
\vspace{-0.3cm}
\end{figure}
As a result, the height of radiosondes may rise rapidly, {which may occur between the outer rainband region and the eyewall region, as depicted in Fig. \ref{fig:3dmodel1}.}
This situation will lead to the concentrating distribution of most radiosondes in the upper layer of the typhoon, that is, the area closer to the RR. Yet, there are also a few radiosondes distributed farther away from the RR. Consequently, the vertical distance $H_i$ between any radiosonde $i$ and the RR may follow an exponential distribution. 
 For the Weibull distribution, we can transform it into an exponential distribution when setting $k_s = 1$ and $\lambda_s = 2$.
 Meanwhile, there is uncertainty about the horizontal moving direction of a radiosonde in this case. We can then utilize Eq. (\ref{eq:irregular_horizontal_model}) to calculate the PDF of the horizontal distance $R_i$. 
Further, these two types of distances are independent of each other.
As a result, the 3D spatial distribution of a radiosonde can be modeled as 
\begin{equation}{\label{case1js}}
     \begin{array}{l}
    {f_{{R_i},{H_i}}}\left( {{r_i},{h_i}} \right) = \\
    \hspace{-2em}\quad \left\{ {\begin{array}{*{20}{l}}
{\frac{\pi {\lambda _n}{r_i}{e^{ - \pi {\lambda _n}r_i^2 - \frac{{{h_i}}}{2}}}}{{[1 - e^{-\pi\lambda_n{(R^{max}_{RX})}^2}][e^{-\frac{H_{RX}^{min}}{2}}-e^{-\frac{H_{RX}^{max}}{2}}]}},}&\begin{array}{l}
{h_i} \in \left[ {H_{RX}^{min},H_{RX}^{max}} \right]\\
\& {r_i} \in \left[ {0,R_{RX}^{max}} \right],
\end{array}\\
{0,}&{else.}
\end{array}} \right.
\end{array}
\end{equation}

This case is motivated by {Eqs. (\ref{coolingeffect}) and (\ref{velocity}). These two formulas indicate that the upward convective activity in a typhoon is enhanced, making it reasonable to assume that the distribution of radiosonde follows the distribution described by Eq. (\ref{case1js}).}
To achieve the goal of comprehensively collecting meteorological data from the interior of a typhoon, all radiosondes are distributed at different $R_i$ and $H_i$ in the 3D spatial distribution model, guaranteeing the uniqueness of each radiosonde's location. 
\subsubsection{Case 2}
\begin{figure}[!t]
\centering
\includegraphics[width=3in]{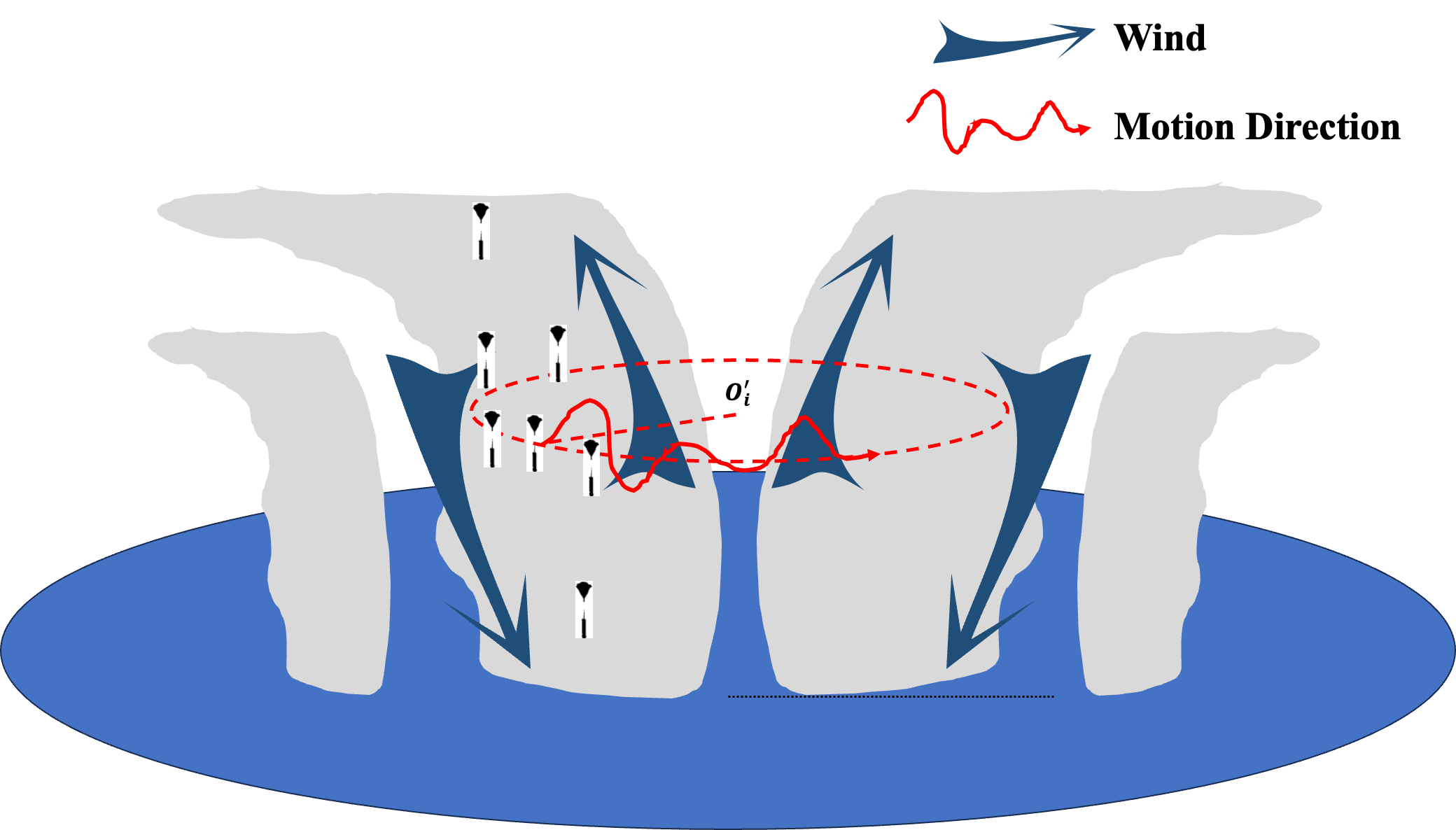}
\caption{Wind and radiosonde motion directions in Case 2.}
\label{fig:3dmodel2}
\vspace{-0.3cm}
\end{figure}
In certain regions within a typhoon, there are simultaneous occurrences of updrafts and downdrafts, causing fluctuations in altitude for radiosondes within certain layers of the typhoon. Due to the Coriolis force, the airflow inside the typhoon exhibits a cyclonic rotation; {at the same time, the friction of water drops or ice pieces against air can decelerate the air updraft \cite{tf5}}, potentially causing the horizontal movement of radiosondes to manifest as circular motion, as illustrated in Fig. \ref{fig:3dmodel2}. {The following three equations characterize the dynamic balance of fluids in certain regions of a typhoon \cite{tf5}.}
{
\begin{equation}{\label{rfriction}}
    \hspace{-1em}\varrho\biggl(\frac{\partial v_{r}}{\partial t}+v_{r} \frac{\partial v_{r}}{\partial r}-\frac{1}{r}v_{0}^{2}+v_{z} \frac{\partial v_{r}}{\partial z}\biggr)=-R_{1} \frac{\partial(\varrho T)}{\partial r}+f_{0}\varrho v_{\vartheta}-\varepsilon_{1}(z)v_{r}
\end{equation}
\begin{equation}{\label{thetafriction}}
    \varrho\biggl(\frac{\partial v_{\vartheta}}{\partial t}+v_{r} \frac{\partial v_{\vartheta}}{\partial r}+\frac{1}{r}v_{\vartheta}v_{r}+v_{z} \frac{\partial v_{\vartheta}}{\partial z}\biggr)=-f_{0}\varrho v_{r}-\varepsilon_{1}(z)v_{\vartheta}
\end{equation}
\begin{equation}{\label{zfriction}}
    \varrho\biggl(\frac{\partial v_{z}}{\partial t}+v_{r} \frac{\partial v_{z}}{\partial r}+v_{z} \frac{\partial v_{z}}{\partial z}\biggr)=-R_{1} \frac{\partial(\varrho T)}{\partial z}-[\Sigma+\varrho]g,
\end{equation}
where, $r$ is the independent variable with respect to the distance from the center of a tropical cyclone, and $z$ is the variable with respect to the altitude. $v_r$, $v_\theta$, and $v_z$ are the radial, tangential, and vertical components of the velocity, respectively. $\varrho$, $T$, $R_1$, and $g$ are respectively the density, temperature, a gas constant, and the acceleration of gravity. $\varepsilon_{1}(z)$ is the coefficient of friction, and $\Sigma$ is the amount of liquid and solid water.}

{$\varepsilon_{1}(z)v_{r}$ and $\varepsilon_{1}(z)v_{\vartheta}$ in Eqs. (\ref{rfriction}) and (\ref{thetafriction}) represent the radial frictional force and tangential frictional force, respectively. The radial frictional force counteracts the acceleration of the radial flow, causing the radial flow of the cyclone to gradually stabilize. Meanwhile, the tangential frictional force weakens the rotation intensity of the air. This deceleration effect, along with the balance of radial velocity, contributes to a more stable cyclone structure. Eq. (\ref{zfriction}) reveals the dynamic balance of vertical air movement within the typhoon, with a focus on the effects of temperature gradients and latent heat release. $-[\Sigma+\varrho]g$ indicates that the upward motion is influenced by the combined gravitational force of the air itself and the liquid and solid water.}
Hence, radiosondes tend to concentrate within a specific altitude range.
We then utilize a normal distribution to model the distribution of the vertical distance. 
The Weibull distribution degenerates to a normal distribution when setting $k_s = 6$ and $\lambda_s = 15$.
Meanwhile, the circular motion model is utilized to model the motion of a radiosonde in the horizontal direction.
{The dynamic balance within the typhoon, as described by Eqs. (\ref{rfriction}), (\ref{thetafriction}), and (\ref{zfriction}), provides an explanation for the distribution and the model, further validating the reasonableness of this case.}
As a result, we write the expression of the 3D spatial distribution of a radiosonde as
\begin{equation}\label{PDF_case2_altitude}
\begin{array}{l}
{f_{{R_i},{H_i}}}\left( {{r_i},{h_i}} \right) = \\
\quad \left\{ {\begin{array}{*{20}{l}}
{\frac{\frac{{4{r_i}}}{{5{{(R_{RX}^{max})}^2}}}{{(\frac{{{h_i}}}{{15}})}^5}{e^{ - {{(\frac{{{h_i}}}{{15}})}^6}}}}{{e^{-(\frac{H_{RX}^{min}}{15})^6}-e^{-(\frac{H_{RX}^{max}}{15})^6}},}}&\begin{array}{l}
{h_i} \in \left[ {H_{RX}^{min},H_{RX}^{max}} \right]\\
\& {r_i} \in \left[ {0,R_{RX}^{max}} \right],
\end{array}\\
{0,}&{else.}
\end{array}} \right.
\end{array}
\end{equation}

Unlike case 1, this case may exhibit {the stable phase of a typhoon's development}. Taking into account the typhoon's internal structure and the impact of wind, the probability increases for radiosondes distributed at particular $H_i$ and $R_i$, {such as the region between the eyewall  and eye of a typhoon}, while it decreases comparatively for other peripheral locations. 
Meanwhile, the theoretical analysis of the communication performance in this case is not a straightforward task. Nevertheless, performance simulations and calculations can be conducted using the Monte Carlo simulation. 

{\subsubsection{Case 3}
Typhoons have a wide range and a complex and dynamic structure. It is possible that a hybrid scenario of Case 1 and Case 2 exists within the typhoon. As is shown in Table \ref{tab:table3}, the discussion of this hybrid case can be approached by integrating both Case 1 and Case 2.
Since the ranges of the variables for the two cases are the same, the spatial distribution of a radiosonde can be given by
\begin{equation}\label{PDF_case3_altitude}
\begin{array}{l}
{f_{{R_i},{H_i}}}\left( {{r_i},{h_i}} \right) = \\
\hspace{-2em}\quad \left\{ {\begin{array}{*{20}{l}}
{p_1{f^{Case 1}_{{R_i},{H_i}}}\left( {{r_i},{h_i}} \right) + p_2{f^{Case 2}_{{R_i},{H_i}}}\left( {{r_i},{h_i}} \right)}&\begin{array}{l}
{h_i} \in \left[ {H_{RX}^{min},H_{RX}^{max}} \right]\\
\& {r_i} \in \left[ {0,R_{RX}^{max}} \right],
\end{array}\\
{0,}&{else.}
\end{array}} \right.
\end{array}
\end{equation}
\hspace{1em}The 3D spatial distribution of a radiosonde in Case 1 and Case 2, respectively, converges to 1 over the value interval. Since the value intervals are the same, that is, the integration intervals are the same, and the sum of $p_1$ and $p_2$ is 1, it is easy to prove that the distribution of Case 3 converges to 1 over the value interval.
}
\begin{table}[!t]
\caption{{Hybrid scenario in Case 3}}
\label{tab:table3}
\centering
\begin{tabular}{|c|c|c|}
\hline
$Case$&$Case 1$&$Case 2$\\
\hline
$Probability$&$p_1$&$p_2$\\
\hline
\multicolumn{3}{|c|}{$\mathbb{P}(Case3)=p_1\mathbb{P}(Case1)+p_2\mathbb{P}(Case2)$} \\
\hline
\end{tabular}
\vspace{-0.3cm}
\end{table}

\section{Connection probability analysis}
We investigate the uplink CP of the radiosonde network. The CP denoted as $p_c$ is defined as the probability that the SINR is greater than a threshold $T_s$, which is expressed as
\begin{equation}
    p_c = \mathbb{P}[SINR > T_s]
\end{equation}
where SINR can be calculated by
\begin{equation*}
    \begin{split}
    {SINR} &= \frac{\mathcal{P}_r}{\sigma^2 + \mathcal{I}_\mathcal{Z}}\\
    &= \frac{\mathcal{P}_t g_{TR} A^{-1}_{r_T} L^{-\alpha}_{TR}H^{\alpha\epsilon}_{TR}g_tg_r}{\sigma^2 + \mathcal{I}_\mathcal{Z}}\\
    &= \frac{\mathcal{P}_t g_{TR} r^{-1}{\gamma_R}^{-1} \frac{sin\theta_{TR}}{H_{TR}-H^{min}_{RX}}  L^{-\alpha}_{TR}H^{\alpha\epsilon}_{TR}g_tg_r}{\sigma^2 + \mathcal{I}_\mathcal{Z}}\\
    \end{split}
\end{equation*}
\begin{equation}
    \begin{split}
        &= \frac{\mathcal{P}_t g_{TR} r^{-1}{\gamma_R}^{-1} L^{-\alpha-1}_{TR}\frac{H^{\alpha\epsilon+1}_{TR}}{H_{TR}-H^{min}_{RX}}g_tg_r}{\sigma^2 + \mathcal{I}_\mathcal{Z}}
    \end{split}
\end{equation}
and the signal interference $\mathcal{I}_\mathcal{Z}$ can be computed by 
\begin{equation}
\begin{split}
    \mathcal{I}_\mathcal{Z} &= \sum_{z \in \mathcal{Z}}{P}_t g_z A^{-1}_{r_z}L^{-\alpha}_{z}H^{\alpha\epsilon}_{z}g_tg_r \\
    &= \sum_{z \in \mathcal{Z}}{P}_t g_z r^{-1}{\gamma_R}^{-1} L^{-\alpha-1}_{z}\frac{H^{\alpha\epsilon+1}_{z}}{H_{z}-H^{min}_{RX}}g_tg_r
\end{split}
\end{equation}

CP can also be visualized as being the average area or the average fraction of radiosondes in coverage. In the previous section, we discussed two distribution cases, each of which is associated with a distinct model. We then theoretically derive the expression of CP.

The SINR expression is so complicated that it is highly challenging to derive the closed-form expression of the uplink CP. 
Fortunately, the stochastic geometric theory can be leveraged to accomplish this hard task, which has been confirmed in many works \cite{ref9,ref10,ref71,ref72,ref73,ref74,ref75,ref76,ref77,ref78,ref79,ref82,ref83,ref85}. 
However, differing from the existing literature \cite{ref71,ref72,ref73}, which utilizes the stochastic geometric theory to analyze the connection performance of a 2D network, this paper discusses the connection performance of a 3D network. Although there exists a number of literature discussing the connection performance of 3D networks, the altitude distribution of network nodes is significantly different from that in this paper, and is much simpler than the altitude distribution designed in this paper. The difference in the altitude distribution makes their derivation processes and approaches quite different. Thus, the existing derivation approaches cannot be directly applied in this paper.

\subsection{CP Expression Derivation in Case 1}
In this case, the random variable $R_i$ and the random variable $H_i$ are independent. From Figs. \ref{fig:enter-label2} and \ref{fig:enter-label3}, we can intuitively observe that $L_i = \sqrt{R^2_i + H^2_i}$, and $L_i \in ({H^{min}_{RX}},{\sqrt{{H^{max}_{RX}}^2 + {R^{max}_{RX}}^2}})$, the CDF and PDF of $L_i$ can then be obtained from the distributions of $R_i$ and $H_i$. Mathematically, we can write the probability of $L_i$ as
\begin{equation}\label{eq17}
    \begin{split}
        \mathcal{F}_{L_i}\left(l_i\right) &= \mathbb{P}\left[L_i  \leq l_i\right] \\
        &= \mathbb{P}\left[\sqrt{R^2_i + H^2_i} \leq l_i\right] \\
        &= \iint_{(u,v) \subset (R_i,H_i)} {f_{R_i,H_i}\left(u,v\right)}dudv
        \\
        &= \int^{R^{max}_{RX}}_0 f_{R_i}(u) \int^{h_i}_{H^{min}_{RX}} f_{H_i}(v)dvdu
        \\
        &= \frac{\int^{R^{max}_{RX}}_0 2\pi\lambda_nue^{-\pi{\lambda_n}u^2}\int^{\sqrt{l^2_i-u^2}}_{H^{min}_{RX}} \frac{1}{2}e^{-\frac{v}{2}}dvdu}{[1 - e^{-\pi\lambda_n{(R^{max}_{RX})}^2}][e^{-\frac{H_{RX}^{min}}{2}}-e^{-\frac{H_{RX}^{max}}{2}}]}
    \end{split}
    \end{equation}

    {However, it is difficult to obtain the closed-form expression of CDF of $L_i$. In (\ref{eq17}), there are both square root and exponential square terms. Besides, there is an improper integration term that cannot be solved by the basic integration formula. To this end, we utilize an equal substitution scheme to transform complex terms and the Taylor formula to deal with the improper integration item. The following lemma presents the closed-form expression of $\mathcal{F}_{L_i}$.}
\begin{lemma}
    {The closed-form expression of the CDF of $L_i$ is expressed as} \\
    \begin{equation}
    \begin{split}
    \mathcal{F}_{L_i}\left(l_i\right) &= C_1[e^{-\frac{{H^{min}_{RX}}}{2}}[1 - e^{-\pi\lambda_n{(R^{max}_{RX})}^2}] - e^{-\frac{l_i}{2}} \\
    & + e^{-\pi\lambda_n(R^{max}_{RX})^2-\frac{\sqrt{l^2_i -(R^{max}_{RX})^2}}{2}}- \\
    &\hspace{-4em}\sum_{n=0}^{\infty}\frac{(\pi\lambda_n)^n}{n!}\frac{({l_i-\frac{1}{4\pi\lambda_n}})^{2n+1}-({\sqrt{l^2_i-(R^{max}_{RX})^2}-\frac{1}{4\pi\lambda_n}})^{2n+1}}{2n+1}]
    \end{split} 
    \end{equation}
{where $C_1 = \frac{1}{[1 - e^{-\pi\lambda_n{(R^{max}_{RX})}^2}][e^{-\frac{H_{RX}^{min}}{2}}-e^{-\frac{H_{RX}^{max}}{2}}]}$.}
\end{lemma}
\begin{proof}
    Please refer to Appendix A.
\end{proof}

{Typically, the PDF of $L_i$ is obtained by differential derivation of the final result of CDF. However, the expression of $\mathcal{F}_{L_i}$ is so complex that the conventional differential approach cannot be directly applied. 
Alternatively, the theorem of differentiation under the integral sign \hspace{0.2em}$\frac{d}{dy}\int^{+\infty}_af(x,y)dx = \int^{+\infty}_a\frac{\partial}{{\partial}y}f(x,y)dx$ is introduced, where differentiation of the integral term is performed in advance to eliminate terms independent of the integration variable. Meanwhile, the equal substitution scheme and the Taylor formula are explored to tackle the exponential term and improper integration. As a result, we can obtain the closed-form expression of the PDF of $L_i$.} The following lemma shows its closed-form expression. 
\begin{lemma}
    The closed-form expression of PDF of $L_i$ takes the following form \\
    \begin{equation}
        \begin{split}\label{case1l_pdf}
        f_{L_i}(l_i) &= C_1\pi\lambda_nl_ie^{-\pi\lambda_nl^2_i-\frac{1}{16\pi\lambda_n}} \times\\
        &\hspace{-3em}\sum_{n=0}^{\infty}\frac{(\pi\lambda_n)^n}{n!}\frac{({l_i-\frac{1}{4\pi\lambda_n}})^{2n+1}-({\sqrt{l^2_i-(R^{max}_{RX})^2}-\frac{1}{4\pi\lambda_n}})^{2n+1}}{2n+1}
        \end{split}
    \end{equation}     
\end{lemma}
\begin{proof}
    Please refer to Appendix B.
\end{proof}

Given the PDF expression of $L_i$, we can rewrite the expression of the uplink CP as
\begin{equation*}
    \begin{split}
        p_c &= \mathbb{P}[\frac{\mathcal{P}_r}{\sigma^2 + \mathcal{I}_\mathcal{Z}} > T_s | L_{TR}] \\
        &= \int^{L^{max}_{TR}}_{L^{min}_{TR}}\mathbb{P}[
        \frac{\mathcal{P}_r}{\sigma^2 + \mathcal{I}_\mathcal{Z}} > T_s]f_{L_{TR}}(x)dx\\
        &= \int^{L^{max}_{TR}}_{L^{min}_{TR}}\mathbb{P}[
        \frac{\mathcal{P}_t g_{TR} r^{-1}{\gamma_R}^{-1} L^{-\alpha-1}_{TR}\frac{H^{\alpha\epsilon+1}_{TR}}{H_{TR}-H^{min}_{RX}}g_tg_r}{\sigma^2 + \mathcal{I}_\mathcal{Z}} > T_s] \\
        &\hspace{18em} \cdot f_{L_{TR}}(x)dx\\
    \end{split} 
\end{equation*}
\begin{equation}
\label{eq:cp_original}
\begin{split}
    &= \int^{L^{max}_{TR}}_{L^{min}_{TR}}\mathbb{P}[g_{TR} > \frac{T_s (\sigma^2 + \mathcal{I}_\mathcal{Z})}{\mathcal{P}_t r^{-1}{\gamma_R}^{-1} L^{-\alpha-1}_{TR}\frac{H^{\alpha\epsilon+1}_{TR}}{H_{TR}-H^{min}_{RX}}g_tg_r}] \\
    &\hspace{18em} \cdot f_{L_{TR}}(x)dx\\
    &= \int^{L^{max}_{TR}}_{L^{min}_{TR}}\mathbb{P}[g_{TR} > \frac{T_s ({\sigma^\prime}^2 + {\mathcal{I}_\mathcal{Z}}^\prime)}{L^{-\alpha-1}_{TR}\frac{H^{\alpha\epsilon+1}_{TR}}{H_{TR}-H^{min}_{RX}}}]f_{L_{TR}}(x)dx\\
\end{split}
\end{equation}
where $L^{max}_{TR} = {\sqrt{{H^{max}_{RX}}^2 + {R^{max}_{RX}}^2}},\hspace{1em} L^{min}_{TR} = {H^{min}_{RX}}$, 
${\sigma^\prime}^2 = \frac{\sigma^2}{\mathcal{P}_tg_tg_rr^{-1}{\gamma_R}^{-1}}$, $\mathcal{I}^\prime_\mathcal{Z} = \frac{\mathcal{I}_\mathcal{Z}}{\mathcal{P}_tg_tg_rr^{-1}{\gamma_R}^{-1}}$. 

{Under the condition that $g_{TR} \sim exp(\mu)$, we can obtain the following result}
\begin{equation}\label{exponentialcal}
    \begin{split}
        \mathbb{P}[g_{TR} > X] &= 1 - \mathbb{P}[g_{TR} \leq X]\\
        &= 1 - \int^X_0 {\mu}e^{-{\mu}g}dg\\
        &= e^{-{\mu}X}
    \end{split}
\end{equation}

With the definition of Laplace transform of interference  $\mathscr{L}_{\mathcal{I}^\prime_\mathcal{Z}}(s) = \mathbb{E}_{\mathcal{I}^\prime_\mathcal{Z}}[e^{-s{\mathcal{I}^\prime_\mathcal{Z}}}]$, (\ref{eq:cp_original}), and (\ref{exponentialcal}), the following theorem gives the expression of the uplink CP in Case 1.
\begin{theorem}\label{uplinkCP}
    The expression of the uplink CP can be written as
    \begin{equation}
    \label{eq:uplink_cp_case1_expression}
    \begin{split}
        p_c(T_s, \lambda_n, \alpha, \epsilon) = \int^{L^{max}_{TR}}_{L^{min}_{TR}} e^{-{\mu}T_s L^{\alpha+1}_{TR}\frac{H_{TR}-H^{min}_{RX}}{H^{\alpha\epsilon+1}_{TR}}{\sigma^\prime}^2} \cdot \\
        \mathscr{L}_{\mathcal{I}^\prime_\mathcal{Z}}[e^{-{\mu}T_s L^{\alpha+1}_{TR}\frac{H_{TR}-H^{min}_{RX}}{H^{\alpha\epsilon+1}_{TR}}}] \cdot f_{L_{TR}}(x)dx
    \end{split}
    \end{equation}
\end{theorem}
where the Laplace transform of the interference is given by
\begin{equation}
\begin{split}
    \mathscr{L}_{\mathcal{I}^\prime_\mathcal{Z}}(s) &= \exp(-2\pi\lambda_n\int^{L^{max}_z}_{l_i}(1 -  \\
&\int^{H^{max}_{RX}}_{H^{min}_{RX}}\frac{\mu}{\mu + sx^{-\alpha-1}\frac{y^{\alpha\epsilon + 1}}{y - H^{min}_{RX}}}\cdot f_{H_z}(y)dy)xdx)
\end{split}
\end{equation}
\begin{proof}
    Please refer to Appendix C.
\end{proof}

Besides, consider that the value of the noise $\sigma^2$ is much smaller than 
the transmit power $\mathcal{P}_t$ of a radiosonde (e.g., $\mathcal{P}_t = 2$ W), we have ${\sigma^\prime}^2 = \frac{\sigma^2}{\mathcal{P}_tg_tg_rr^{-1}{\gamma_R}^{-1}} \approx 0$ and then can obtain the upper bound of CP. 
    The upper bound of the uplink CP can be written as 
    \begin{equation}\label{cp_result}
    \begin{split}
        p_c(T_s, \lambda_n, \alpha, \epsilon) &= \\
        &\hspace{-1em}\int^{L^{max}_{TR}}_{L^{min}_{TR}}\mathscr{L}_{\mathcal{I}^\prime_\mathcal{Z}}[e^{-{\mu}T_s L^{\alpha+1}_{TR}\frac{H_{TR}-H^{min}_{RX}}{H^{\alpha\epsilon+1}_{TR}}}] \cdot f_{L_{TR}}(x)dx
    \end{split}
    \end{equation}

\subsection{CP Expression Derivation in Case 2}
In this case, all radiosondes maintain circular motion of fixed radius in the typhoon environment. With (\ref{PDF_Ri}) and (\ref{PDF_case2_altitude}), we can derive the CDF and PDF of $L_i$ in a similar way to Case 1, and the results are shown as below.
\begin{equation}\label{eq25}
    \begin{split}
        \mathcal{F}_{L_i}\left(l_i\right) &= \frac{\int^{R^{max}_{RX}}_0 \frac{2u}{(R^{max}_{RX})^2}\int^{\sqrt{l^2_i-u^2}}_{H^{min}_{RX}} \frac{2}{5}(\frac{v}{15})^5e^{-(\frac{v}{15})^6} dvdu}{e^{-(\frac{H_{RX}^{min}}{15})^6}-e^{-(\frac{H_{RX}^{max}}{15})^6}}\\
        &= C_2\int^{R^{max}_{RX}}_0 \frac{2u}{(R^{max}_{RX})^2}(e^{-(\frac{H^{min}_{RX}}{15})^6} - e^{-\frac{(l^2_i-u^2)^3}{15^6}})du \\
        &= C_2\int^{R^{max}_{RX}}_0 \frac{1}{(R^{max}_{RX})^2}(e^{-(\frac{H^{min}_{RX}}{15})^6} - e^{-\frac{(l^2_i-u^2)^3}{15^6}})d(u^2) \\
        &\xlongequal{u^2=t} \frac{C_2}{(R^{max}_{RX})^2}\int^{(R^{max}_{RX})^2}_0 [e^{-(\frac{H^{min}_{RX}}{15})^6} - e^{-\frac{(l^2_i-t)^3}{15^6}}]dt
    \end{split}
\end{equation}
{where $C_2 = \frac{1}{e^{-(\frac{H_{RX}^{min}}{15})^6}-e^{-(\frac{H_{RX}^{max}}{15})^6}}$.} In the integration process of Case 2, there is also an improper integration term, which can be tackled by the Taylor formula. The following lemma presents the closed-form expression.

\begin{lemma}
    The closed-form expression of the CDF of $L_i$ is expressed as
    \begin{equation}
    \begin{split}
        &\mathcal{F}_{L_i}\left(l_i\right) = \frac{C_2}{(R^{max}_{RX})^2}[(R^{max}_{RX})^2e^{-(\frac{H^{min}_{RX}}{15})^6}-\\
        &\sum^{\infty}_{n=0} \frac{(-1)^{n}}{n!15^{6n}\cdot(3n+1)}[(l^2_i)^{3n+1} - (l^2_i - (R^{max}_{RX})^2)^{3n+1}]]
    \end{split}
    \end{equation}
\end{lemma}
\begin{proof}
    Please refer to Appendix D.
\end{proof}

Similarly, it is not advisable to obtain the PDF of $L_i$ by directly differentiating the CDF. The derivative theorem under the integral sign is utilized to obtain the integral formula for PDF. The following lemma presents the closed-form expression.
\begin{lemma}
    The closed-form expression of PDF of $L_i$ is given by
    \begin{equation}
    \label{eq:case2_pdf}
    f_{L_i}(l_i) = \frac{2C_2l_i}{(R^{max}_{RX})^2} (e^{-\frac{(l^2_i-(R^{max}_{RX})^2)^3}{15^6}} - e^{-(\frac{l_i}{15})^6})
    \end{equation}
\end{lemma}
\begin{proof}
    Please refer to Appendix E.
\end{proof}

Finally, given (\ref{eq:case2_pdf}), the similar derivation method of the uplink CP in Case 1 can be leveraged to obtain the uplink CP in Case 2. For brief, we omit the detailed derivation processes. The expression of the uplink CP in Case 2 can be obtained by substituting (\ref{eq:case2_pdf}) into (\ref{eq:uplink_cp_case1_expression}).
{\subsection{CP Expression Derivation in Case 3}}
{In this case, given the definition expression of CP and the related probability distribution, the CP can be given by}
{
\begin{equation}
\label{eq:cp_originalcase3}
\begin{split}
    p_c &= \mathbb{P}[\frac{\mathcal{P}_r}{\sigma^2 + \mathcal{I}_\mathcal{Z}} > T_s | L_{TR}] \\
    &\hspace{-1em}= \mathbb{P}[\frac{\mathcal{P}_r}{\sigma^2 + \mathcal{I}_\mathcal{Z}} > T_s | L^{Case1}_{TR}, p_1] + \mathbb{P}[\frac{\mathcal{P}_r}{\sigma^2 + \mathcal{I}_\mathcal{Z}} > T_s | L^{Case2}_{TR}, p_2]  \\
    &\hspace{-1em}= p_1\cdot p^{Case1}_c(T_s, \lambda_n, \alpha, \epsilon) + p_2\cdot p^{Case2}_c(T_s, \lambda_n, \alpha, \epsilon)\\
\end{split}
\end{equation}
\hspace{1em}
Based on 
the derived CP expression (\ref{eq:uplink_cp_case1_expression}) and 
the direct propagation distance PDF of Case 1 and Case 2 (i.e., (\ref{case1l_pdf}) and (\ref{eq:case2_pdf})), the CP in Case 3 can be calculated.} 
\\

\section{Numerical Results}
In this section, we validate the analytical expressions via numerical simulations. Particularly, the numerical results are provided to assess the connection performance of the radiosonde network and verify the two proposed 3D models. The impacts of three key parameters, including power control factor, path loss exponent, and distribution intensity, on the CP are extensively discussed. Some useful insights are also revealed. The default parameter setting is summarized in Table \uppercase\expandafter{\romannumeral3}.

{The numerical simulations were conducted on MATLAB R2021a, which was used for the modeling and computation. The simulations also ran on a computer equipped with an Intel Core i7-10750H processor and 16 GB of memory, with Windows 10 as the operating system. 
In the simulations, the maximum horizontal distance within the communication range of RR is assumed to be 20 km, and the maximum and minimum vertical distance are respectively 20 km and 5 km. The transmission power of radiosondes is set at 33 dBm. The channel model adopts both Rayleigh fading and rain attenuation models. The connection probability of the network was simulated using the derived formulas and the Monte Carlo method. To ensure the stability of the results, each set of simulations using the Monte Carlo method was run 20,000 times, with the average value taken as the final result.}

Fig. \ref{fig:pdf} illustrates the PDF curves of both propagation distance and vertical distance in two cases. From this figure, we can observe that the PDF curves for vertical and propagation distances are similar. It shows that the propagation distance has a strong correlation with the vertical distance and a weak correlation with the horizontal distance. {Next, we discuss the impact of three parameters, including distribution density, path loss exponent, and power control factor, on the CP of the network. At the same time, a comparison is made between the model in this paper and the 3D spherical distribution model presented in \cite{new5}.}
{\subsection{The Impact of Parameter $\lambda$ on CP}}

\begin{table}[!t]
\caption{{Values of some parameters}}
\label{tab:table2}
\centering
\begin{tabular}{|c|c|}
\hline
$\lambda_n = \lambda_r = 0.01, 0.05/km^3$ & $\mu = 1$\\
\hline
$\alpha = 2.0, 4.0$ &$T_s = [-40 : 2 : 0]dB$\\
\hline
$\epsilon = 0, 0.5$ & $R^{max}_{RX} = 20km$\\
\hline
$\sigma^2 = -174dBm/Hz * 10MHz$ & $H^{max}_{RX} = 20km$\\
\hline
$\mathcal{P}_t = 2W = 33dBm$ & $H^{min}_{RX} = 5km$\\
\hline
$g_t = g_r = 2dB$ & $f = 400MHz$\\
\hline
\end{tabular}
\end{table}

\begin{figure}[!t]
\centering
\includegraphics[width=2.5in]{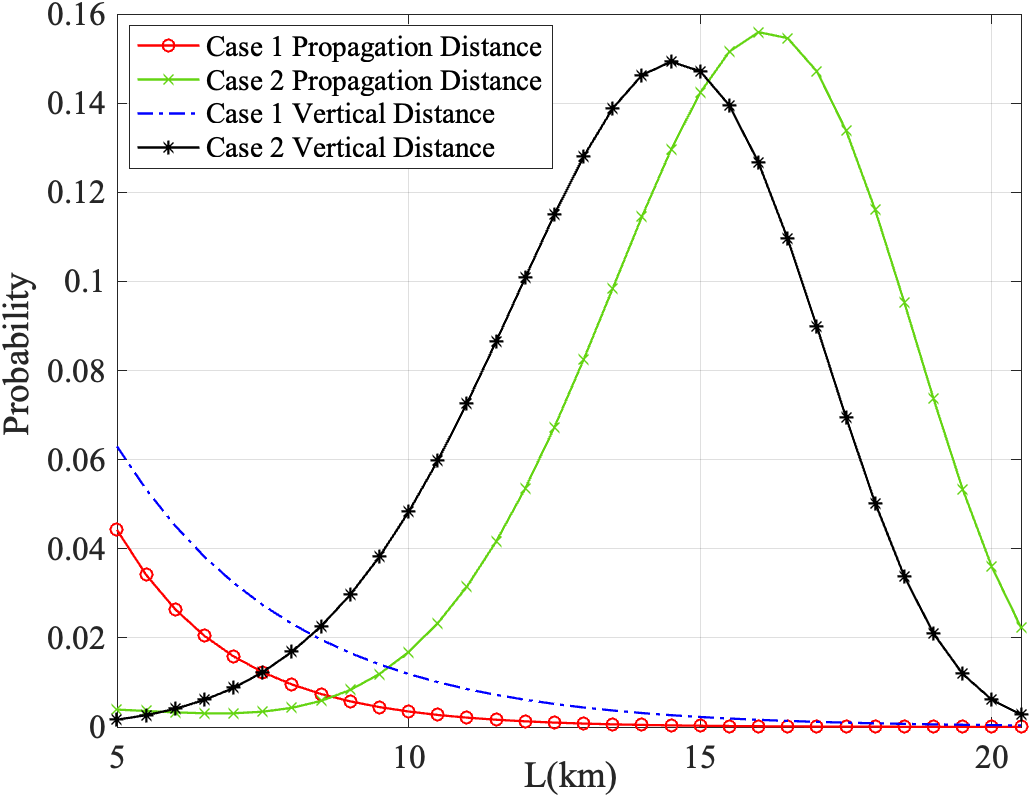}
\caption{The PDF of $H_i$ and $L_i$.}
\label{fig:pdf}
\vspace{-0.3cm}
\end{figure}

{In the scenario to be analyzed, determining the network structure is of primary importance, that is, confirming the number of radiosondes in the network, which is closely related to distribution density. The distribution density of radiosondes refers to the number of radiosondes per unit volume in the typhoon environment. Changes in distribution density will directly affect the structure of the network, which in turn impacts the CP of the network.}

\begin{figure}[!t]
\centering
\includegraphics[width=3in]{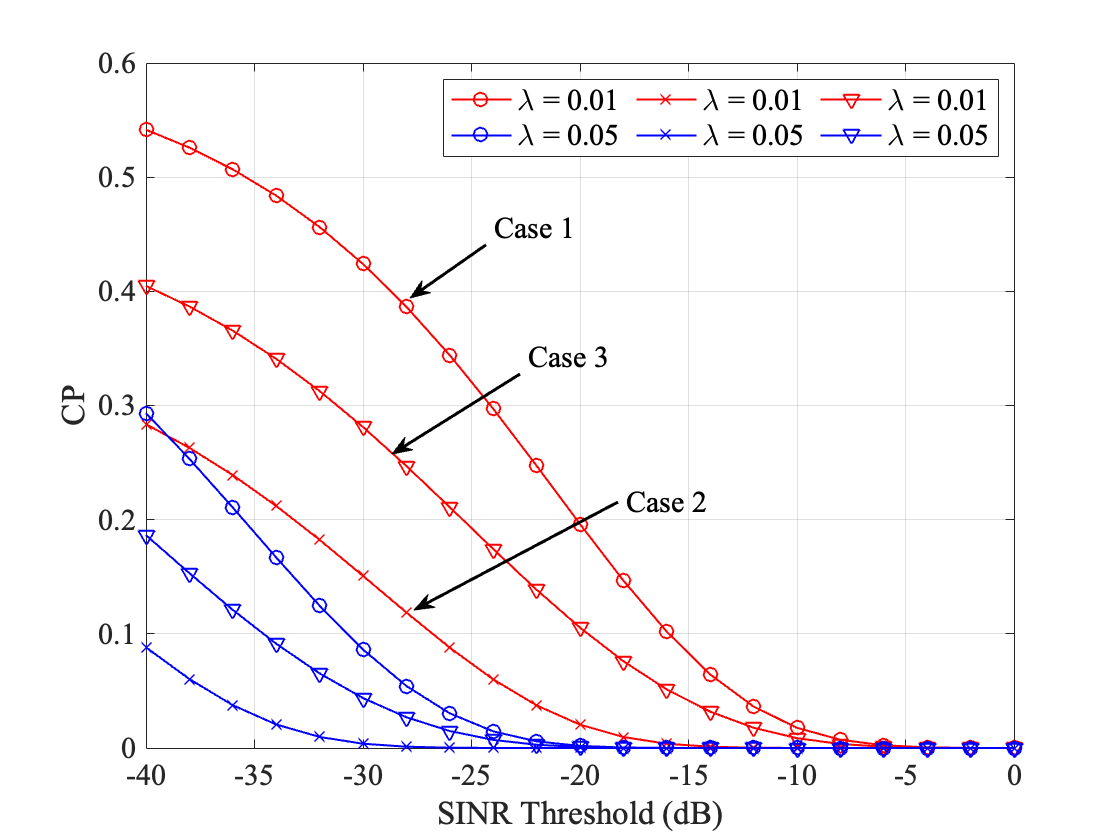}
\caption{The tendency of derived CP vs. SINR threshold \& $\lambda$.}
\label{fig:case123lambdaformula}
\vspace{-0.3cm}
\end{figure}

\begin{figure}[!t]
\centering
\includegraphics[width=3in]{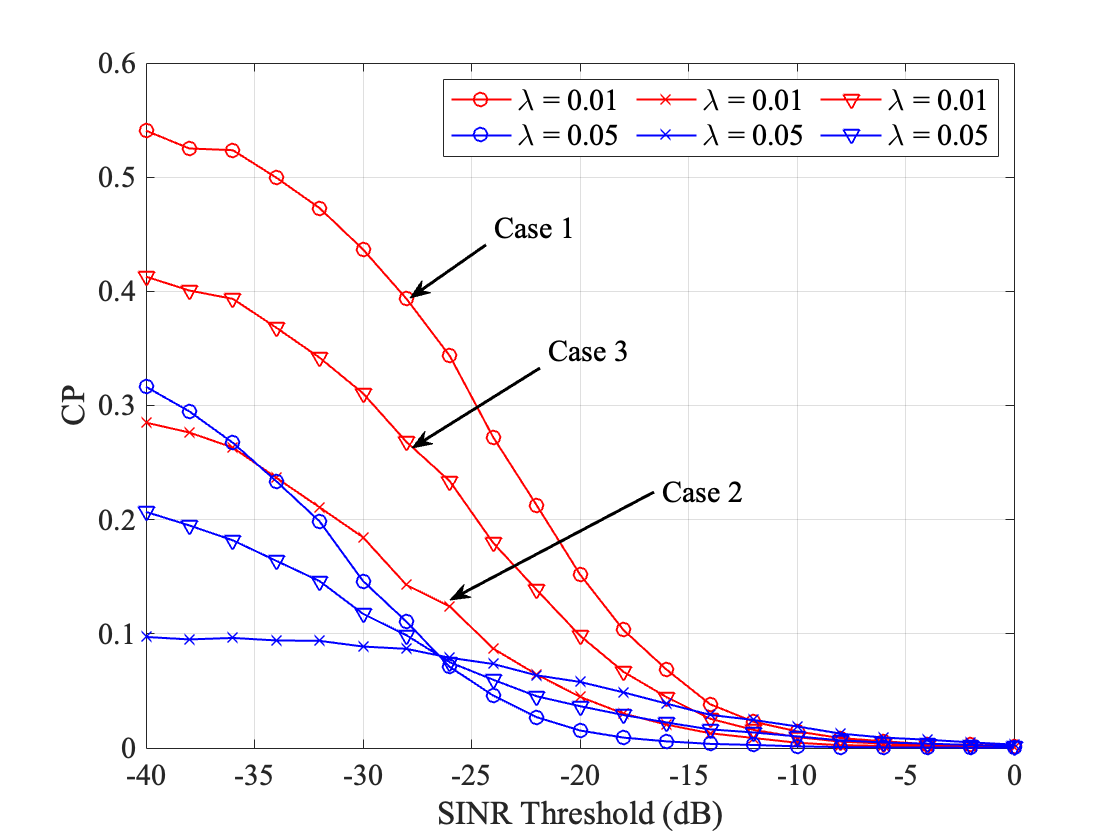}
\caption{The tendency of CP obtained by the Monte Carlo simulation vs. SINR threshold.}
\label{fig:case123lambdamc}
\vspace{-0.3cm}
\end{figure}

\begin{figure}[!t]
\centering
\includegraphics[width=3in]{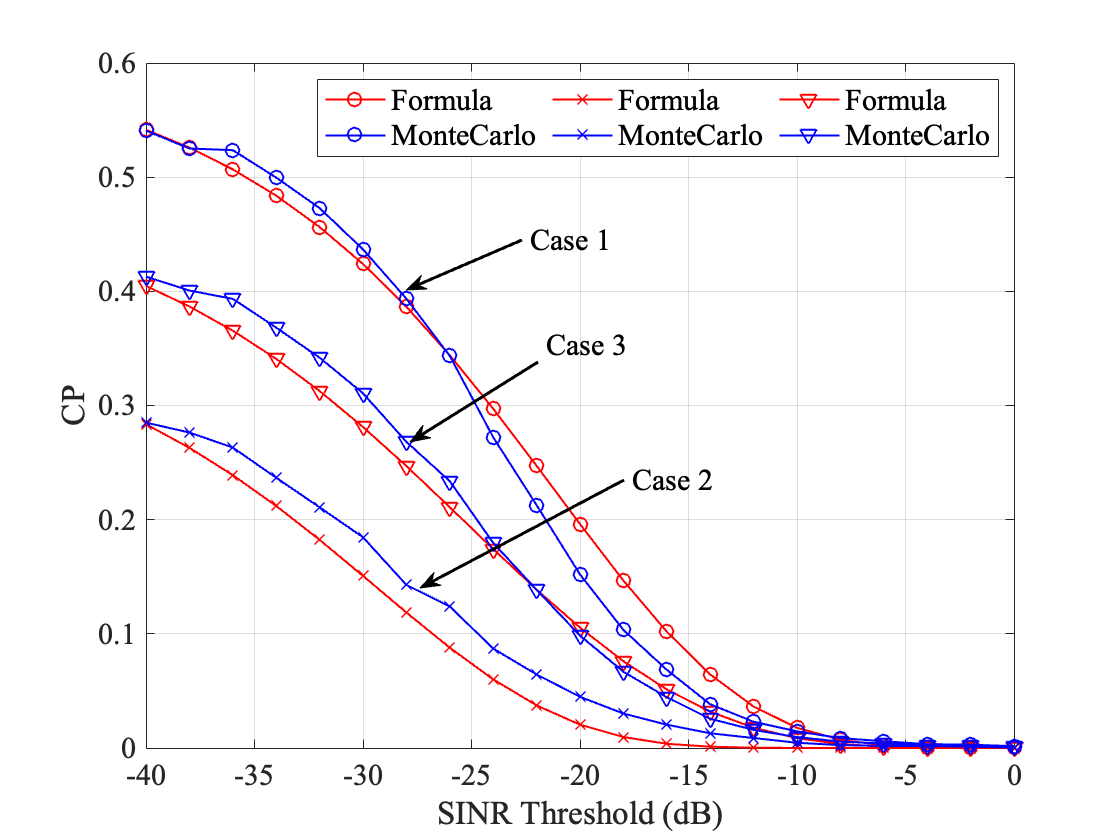}
\caption{A comparison of CP obtained by the theoretical derivation and the Monte Carlo simulation in three cases ($\lambda$ = 0.01).}
\label{fig:case3lambdac}
\vspace{-0.3cm}
\end{figure}

\begin{figure}[!t]
\centering
\includegraphics[width=3in]{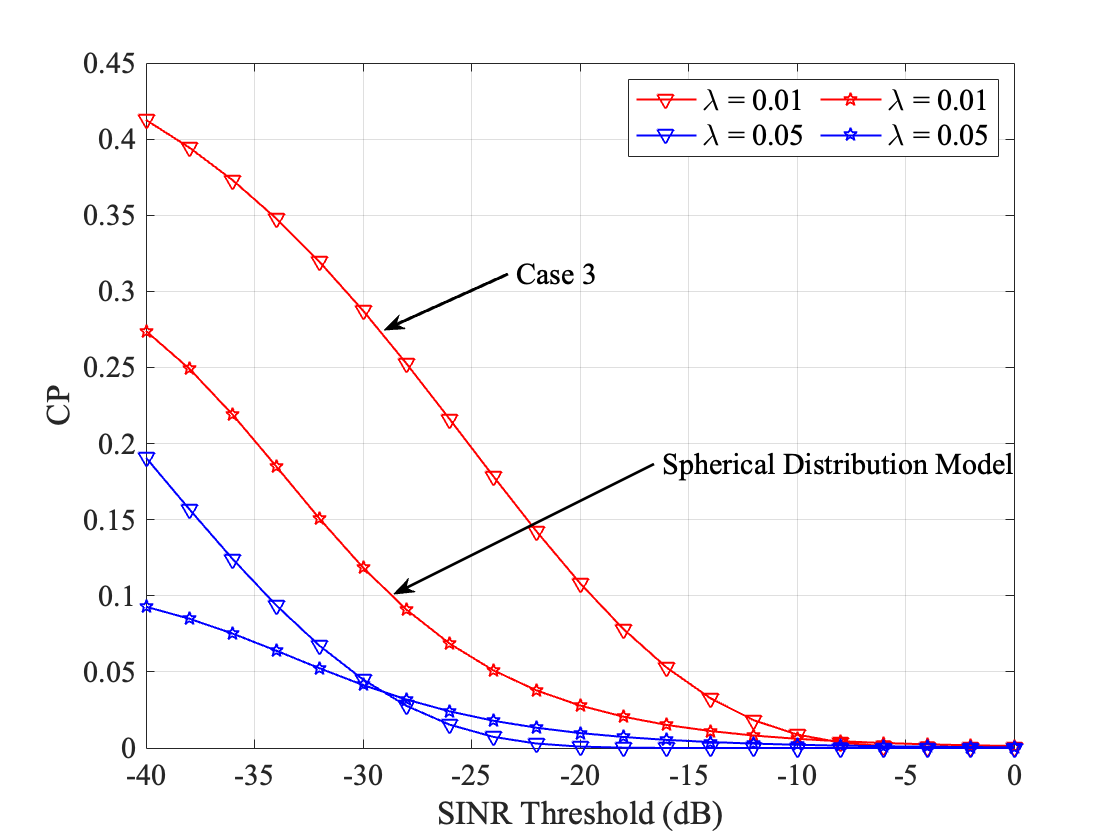}
\caption{A comparison of CP obtained by the theoretical derivation between Case 3 and spherical distribution model.}
\label{fig:case3spherelambdaf}
\vspace{-0.3cm}
\end{figure}

\begin{figure}[!t]
\centering
\includegraphics[width=3in]{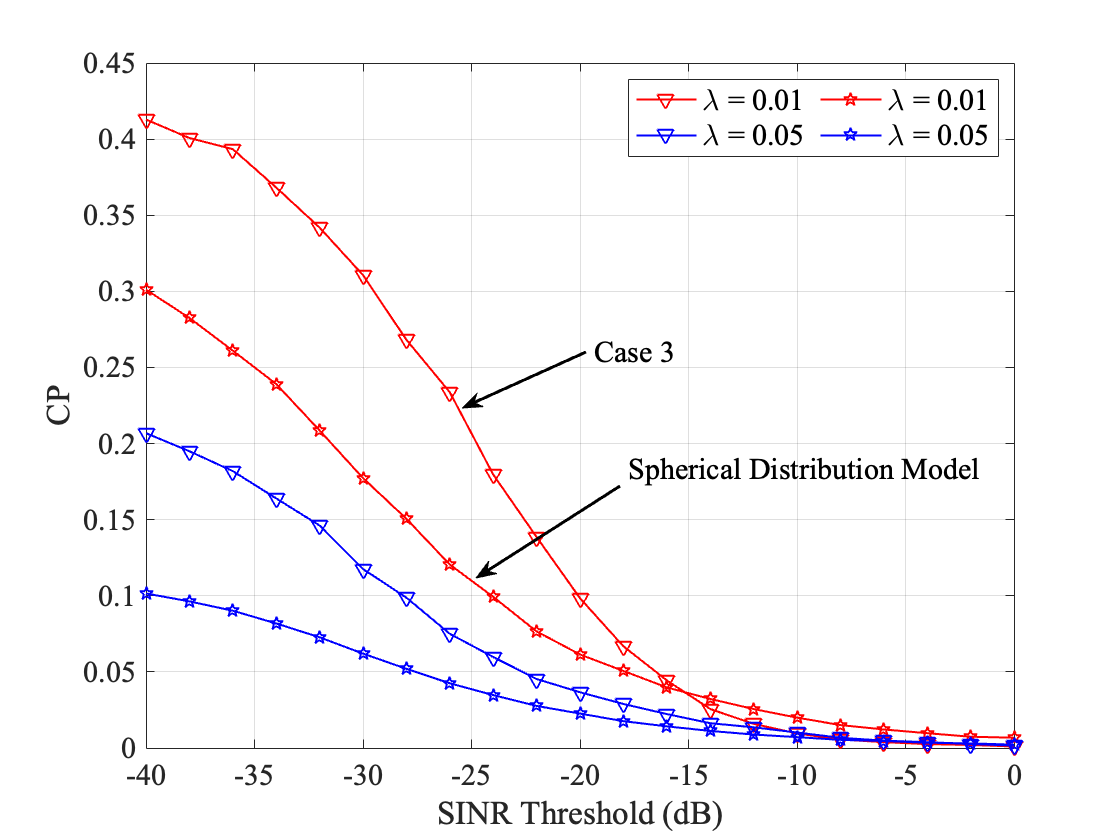}
\caption{A comparison of CP obtained by the Monte Carlo simulation between Case 3 and spherical distribution model.}
\label{fig:case3spheremlambdam}
\vspace{-0.3cm}
\end{figure}

\begin{figure}[!t]
\centering
\includegraphics[width=3in]{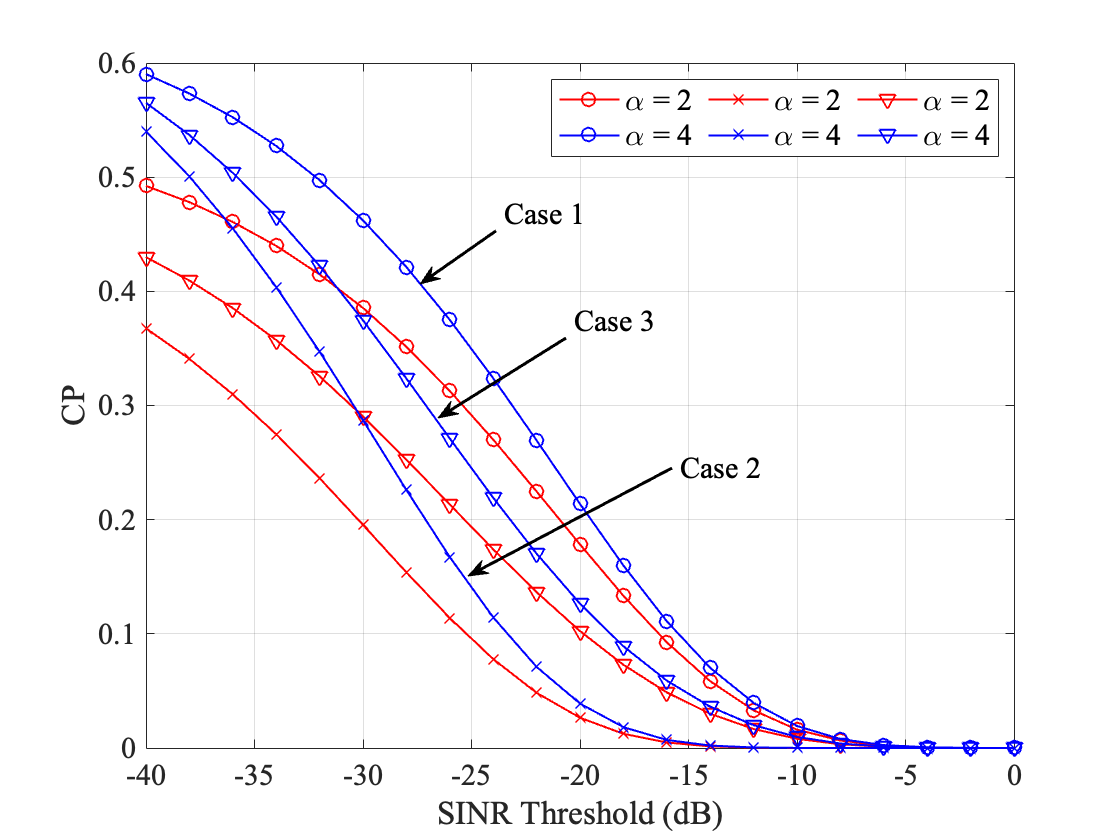}
\caption{The tendency of derived CP vs. SINR threshold \& $\alpha$.}
\label{fig:case123alphaformula}
\vspace{-0.3cm}
\end{figure}
Figs. \ref{fig:case123lambdaformula} and \ref{fig:case123lambdamc} illustrate how variations in the radiosonde density affect the uplink CP. The higher value of $\lambda$ leads to the smaller value of the CP, and the CP in Case 1 is less affected than that in Case 2 and Case 3. A larger radiosonde density means more radiosondes in the same area, which leads to greater interference in SINR. The randomly selected TR will be greatly influenced when it is in the concentrated distribution area of radiosondes. 

{In addition, the results indicate that the distribution in Case 2 is denser, leading to more severe interference affecting the analyzed TR, regardless of the value of $\lambda$. This can also explain that during the development phase of a typhoon, the radiosondes are more dispersed due to strong airflow influences, while during the stable phase, they may concentrate in a particular area, moving with the weak airflow.} Fig. \ref{fig:case3lambdac} shows the comparison results of Monte Carlo simulation and theoretical derivation. The trends of curves obtained by Monte Carlo simulation and theoretical derivation are the same. The results indicate that the values of CP obtained by the two methods are quite similar, especially in Case 3. This trend is also reflected in {Figs. \ref{fig:case3alphacompare} and \ref{fig:case3epsilonc}, suggesting} that the motion patterns and distribution within a typhoon cannot be fully described by a single case. Rather, the value of CP obtained through the combination of various cases may better represent network performance.

{Figs. \ref{fig:case3spherelambdaf} and \ref{fig:case3spheremlambdam} present the curves of CP with Case 3 and spherical distribution model when varying SINR threshold and $\lambda$. The results in Case 3 demonstrate advantages over those derived from the spherical distribution model. In the spherical distribution model with RR as the center, the volume of the region and distribution density dictate the number of radiosondes in a typhoon. Given that RR lies above the typhoon, radiosondes are confined to the lower hemisphere. This spatial limitation, combined with a fixed number of radiosondes, results in a more concentrated distribution, which heightens interference. As a result, the CP value calculated with the same parameters is lower compared to Case 3. This effect is similarly observed in subsequent analyses concerning the impact of the path loss exponent and power control factor on CP. The detailed results can be found in Figs. \ref{fig:case3spherealphaf}, \ref{fig:case3spherealpham}, \ref{fig:case3spheref}, and \ref{fig:case3SphereM}.}

\begin{figure}[!t]
\centering
\includegraphics[width=3in]{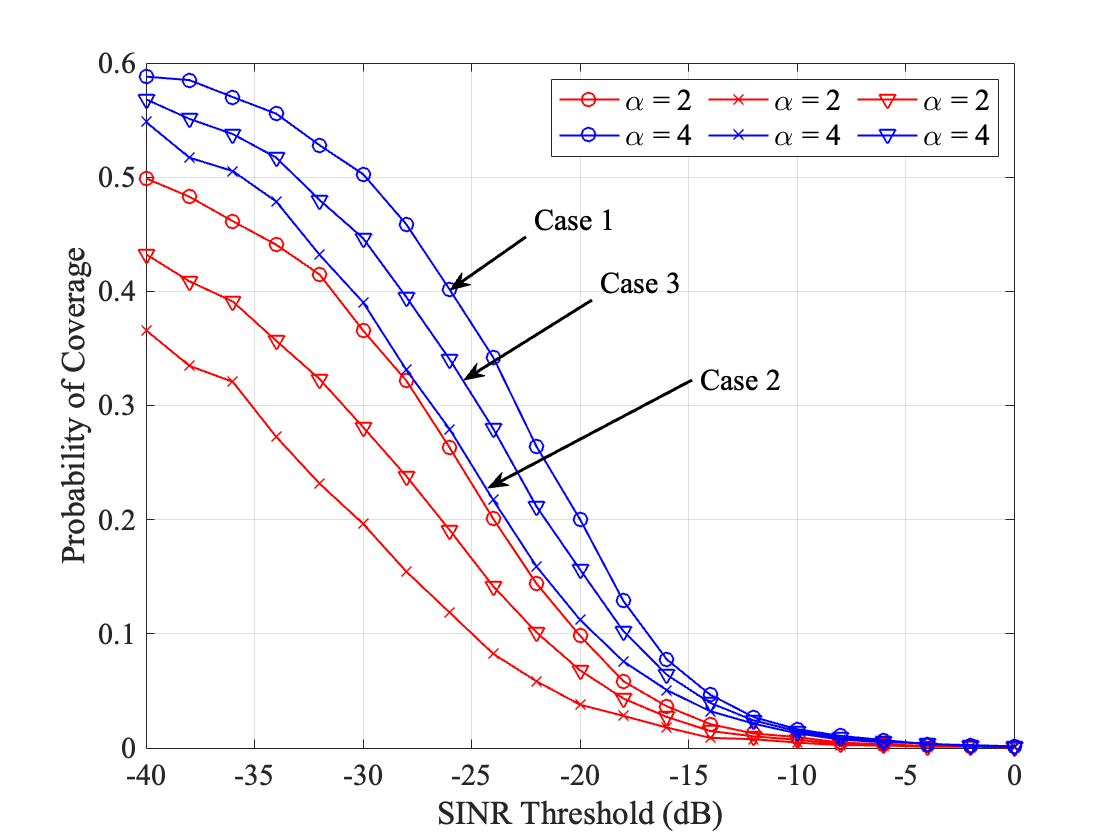}
\caption{The tendency of CP obtained by the Monte Carlo simulation vs. SINR threshold \& $\alpha$.}
\label{fig:case123alphamonte}
\vspace{-0.3cm}
\end{figure}

\begin{figure}[!t]
\centering
\includegraphics[width=3in]{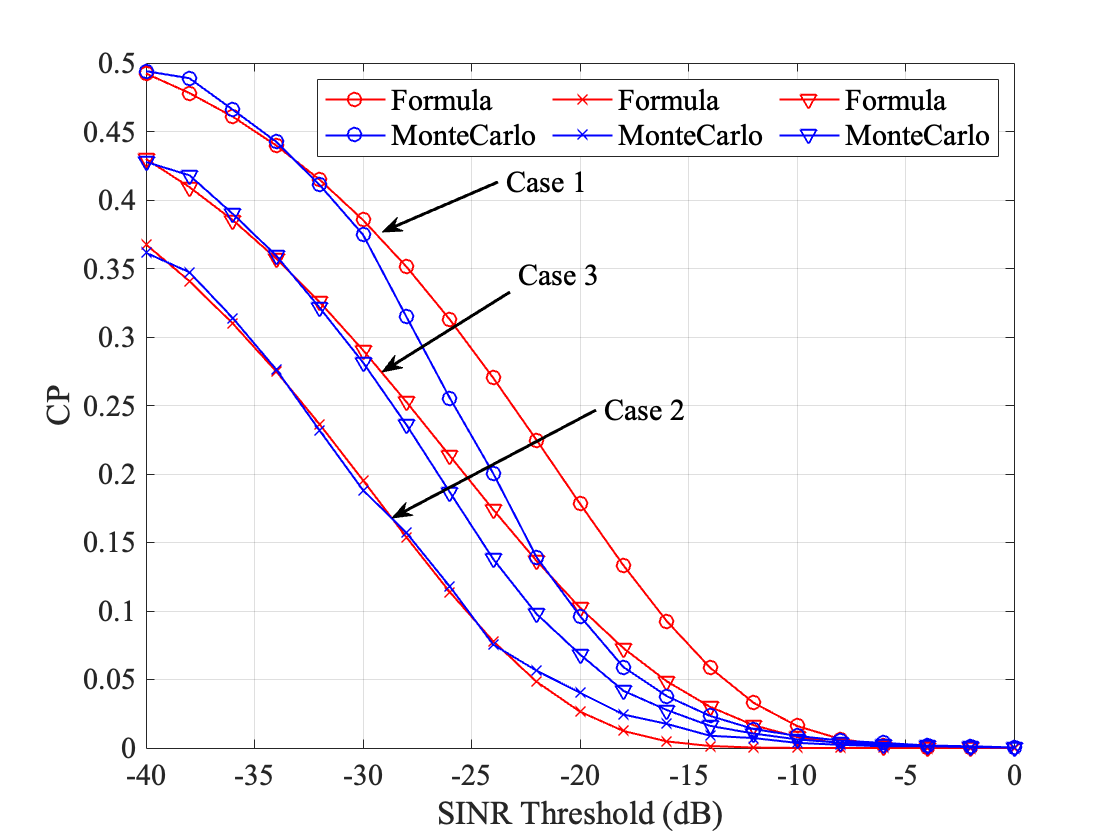}
\caption{A comparison of CP obtained by the theoretical derivation and the Monte Carlo simulation when varying SINR threshold in three cases ($\alpha = 2$).}
\label{fig:case3alphacompare}
\vspace{-0.3cm}
\end{figure}

{\subsection{The Impact of Parameter $\alpha$ on CP}}
{The path loss exponent describes the rate at which a signal decays over distance. When the path loss exponent is elevated, the signal experiences rapid attenuation during propagation, which implies that both the useful signal from the TR and the interference signals from other IRs will also decay quickly.} 

Figs. \ref{fig:case123alphaformula} and \ref{fig:case123alphamonte} depict the influence of variations in the path loss exponent on the CP in three cases. The observed patterns of change are consistent that the higher value of $\alpha$ results in the higher value of CP. Given that the TR is selected at random, the spatial propagation distance from TR to RR may not be the shortest. This means that some IRs are likely to be closer to RR than the TR. Due to the increasing value of $\alpha$, the value of $L_i^{-(\alpha+1)}$ will be smaller, which means greater path loss. Interference signals that are farther away attenuate more quickly, whereas those closer to RR experience slower decay. {\hspace{0.2em}The value of $\frac{L_{TR}^{-(\alpha+1)}}{\sum_{z \in \mathcal{Z}} L_z^{-(\alpha+1)}}$ cannot explain the increasing tendency of CP since it is impossible to distinguish which term decays faster. Instead, we apply theoretical analysis to investigate this phenomenon and derive the Eq. (\ref{cp_result}) for simplicity.}

\begin{figure}[!t]
\centering
\includegraphics[width=3in]{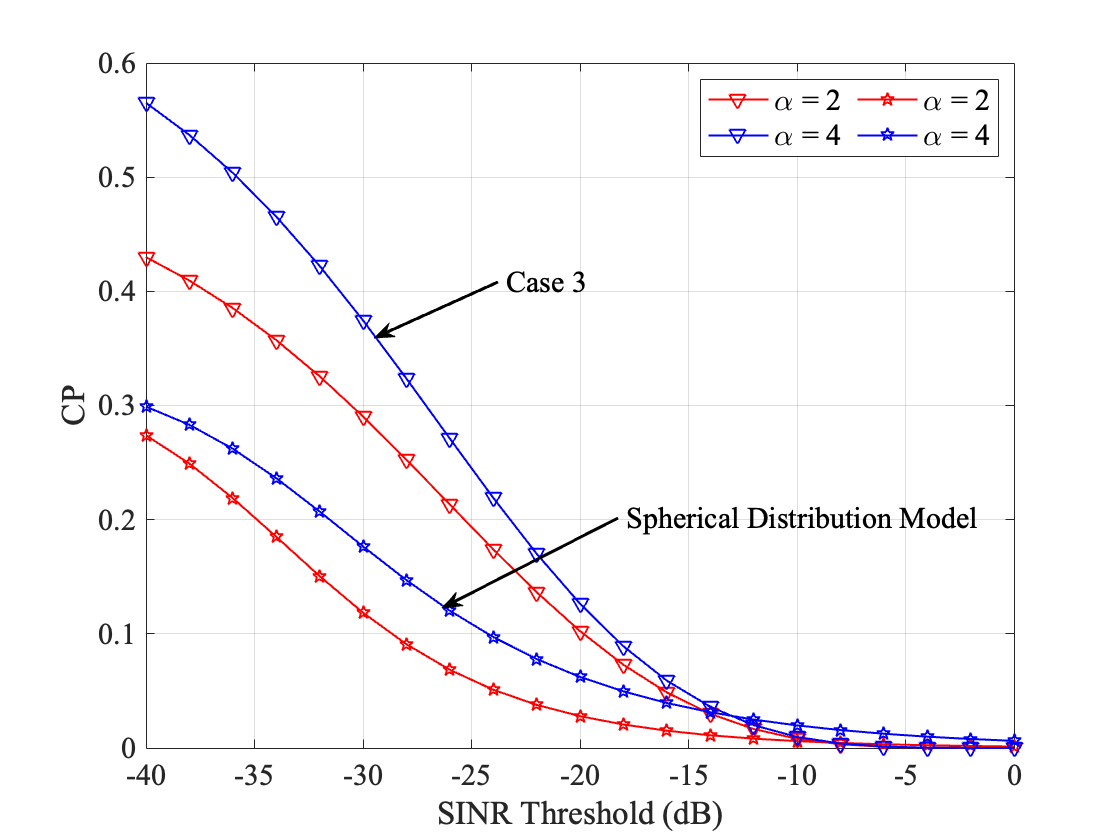}
\caption{A comparison of CP obtained by the theoretical derivation between Case 3 and spherical distribution model.}
\label{fig:case3spherealphaf}
\vspace{-0.3cm}
\end{figure}

\begin{figure}[!t]
\centering
\includegraphics[width=3in]{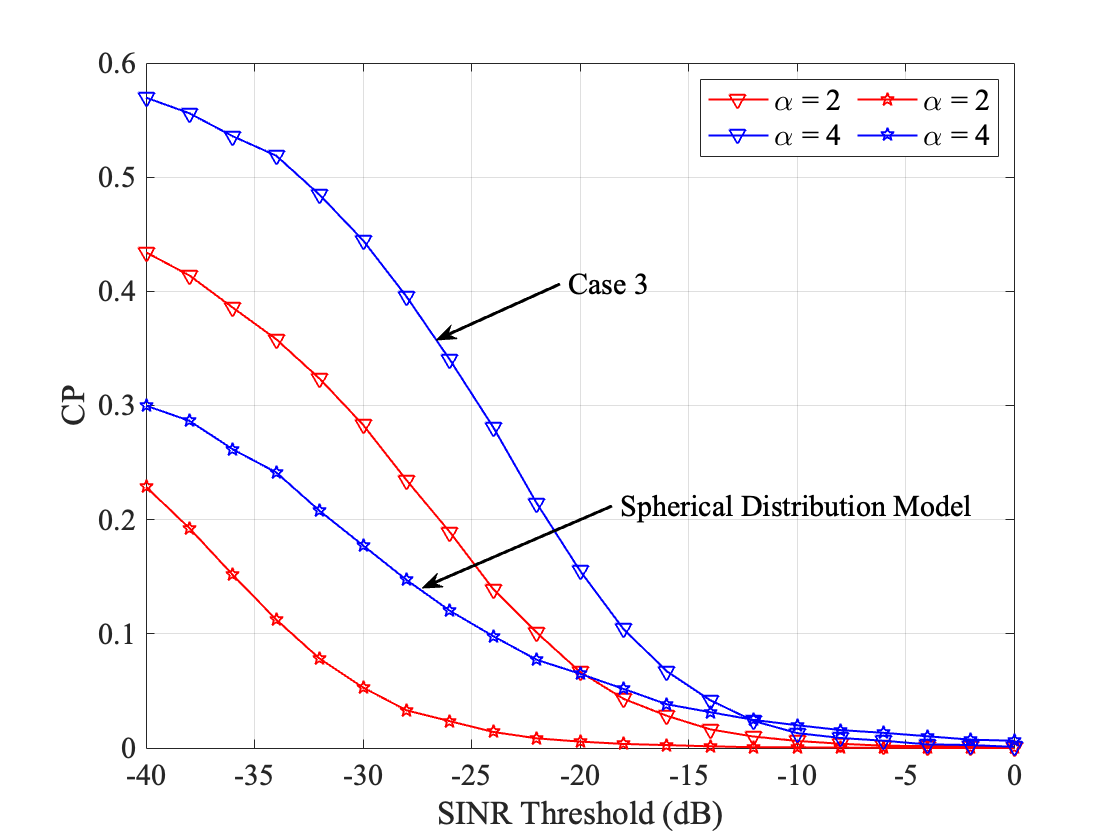}
\caption{A comparison of CP obtained by the Monte Carlo simulation between Case 3 and spherical distribution model.}
\label{fig:case3spherealpham}
\vspace{-0.3cm}
\end{figure}

{To analyze how changes in $\alpha$ affect the CP value, we can simplify the process by treating all other parameters as constants. Then, we calculate the derivative of CP with respect to $\alpha$. The sign of this derivative indicates whether the CP value increases or decreases. In Eq. (\ref{cp_result}), }
\begin{equation}
    \begin{split}
        \mathscr{L}_{\mathcal{I}^\prime_\mathcal{Z}}[&e^{-{\mu}TL^{\alpha+1}_{TR}\frac{H_{TR}-H^{min}_{RX}}{H^{\alpha\epsilon+1}_{TR}}}] = \\
    &\frac{\mu}{\mu + e^{-{\mu}TL^{\alpha+1}_{TR}\frac{H_{TR}-H^{min}_{RX}}{H^{\alpha\epsilon+1}_{TR}}}x^{-\alpha-1}\frac{y^{\alpha\epsilon + 1}}{y - H^{min}_{RX}}} 
    \end{split}
\end{equation}
{is the only term that includes $\alpha$. After further simplification, we can obtain $\frac{1}{1 + e^{-x^{\alpha+1}}x^{-\alpha-1}}$ by assuming $\epsilon = 0$ , $\mu = 1$, other terms to be 1 and $x$ to be constant. The process of derivative is shown below.}
\begin{equation}
    \begin{split}
        \frac{\partial}{\partial \alpha}\left(\frac{1}{1 + e^{-x^{\alpha+1}}x^{-\alpha-1}}\right) &= -\frac{\frac{\partial}{\partial \alpha}\left(e^{-x^{\alpha+1}}x^{-\alpha-1}\right)}{\left(1 + e^{-x^{\alpha+1}}x^{-\alpha-1}\right)^2} \\
        &\hspace{-10em}= - \frac{-e^{-x^{\alpha+1}}x^{\alpha+1}\ln{x} \cdot x^{-\alpha-1} + e^{-x^{\alpha+1}} \cdot x^{-\alpha-1}\left(-\ln{x}\right)}{\left(1 + e^{-x^{\alpha+1}}x^{-\alpha-1}\right)^2} \\
        &\hspace{-10em}= \frac{e^{-x^{\alpha+1}}\ln{x} + e^{-x^{\alpha+1}} \cdot x^{-\alpha-1}\ln{x}}{\left(1 + e^{-x^{\alpha+1}}x^{-\alpha-1}\right)^2} > 0
    \end{split}
\end{equation}

{It is obvious that the value of the derivative is greater than 0. This term exhibits an increasing tendency. Similarly, $-2\pi\lambda_n\int^{L^{max}_z}_{l_i}(1 - \int^{H^{max}_{RX}}_{H^{min}_{RX}}\frac{\mu}{\mu + sx^{-\alpha-1}\frac{y^{\alpha\epsilon + 1}}{y - H^{min}_{RX}}}\cdot f_{H_z}(y)dy)xdx$ has a monotonically increasing property. Ultimately, an increase in the value of $\alpha$ will lead to an increase in the value of CP.}

In Fig. \ref{fig:case3alphacompare}, we compare the simulation results of Case 3 via formula derivation and Monte Carlo simulation. With the same parameters, the CP of the Monte Carlo simulation is slightly smaller than that of the theoretical derivation. It indicates the reasonableness of the formula derivation, and the formula is referential. {In Figs. \ref{fig:case3spherealphaf} and \ref{fig:case3spherealpham}, the results show the better network performance of Case 3.}

\begin{figure}[!t]
\centering
\includegraphics[width=3in]{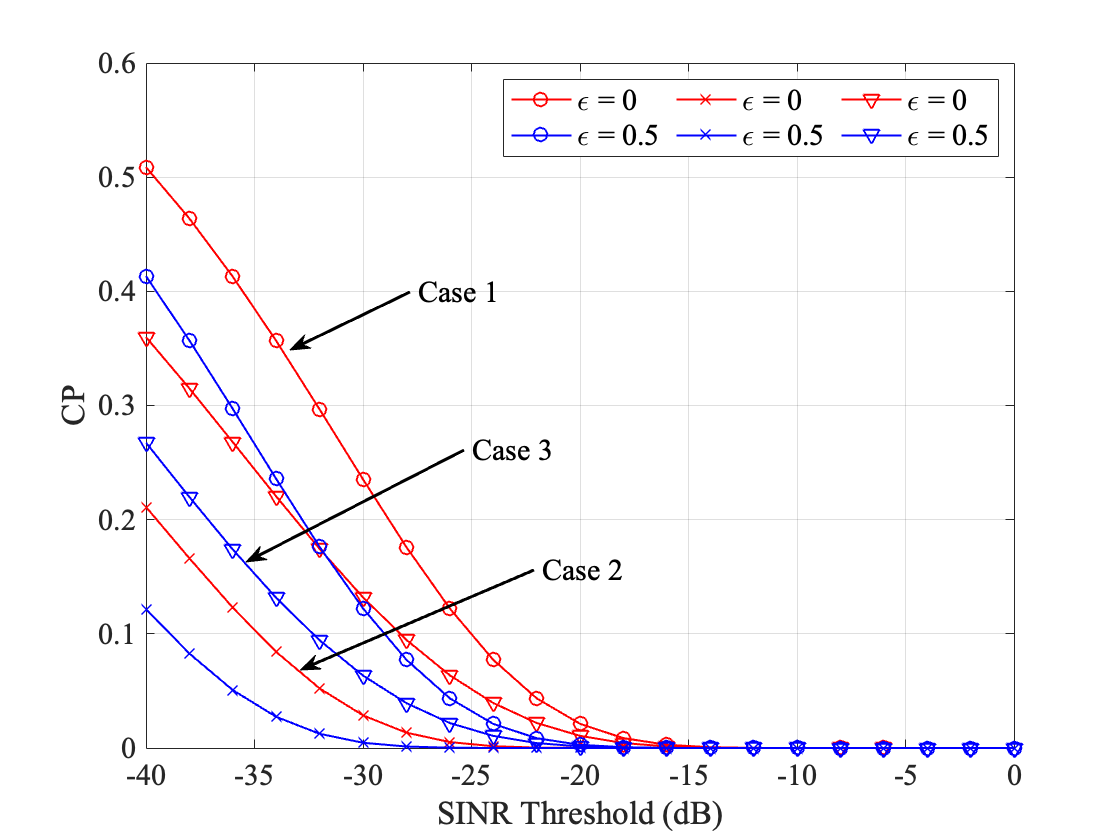}
\caption{The tendency of derived CP vs. SINR threshold \& $\epsilon$.}
\label{fig:case123epsilonf}
\vspace{-0.3cm}
\end{figure}

\begin{figure}[!t]
\centering
\includegraphics[width=3in]{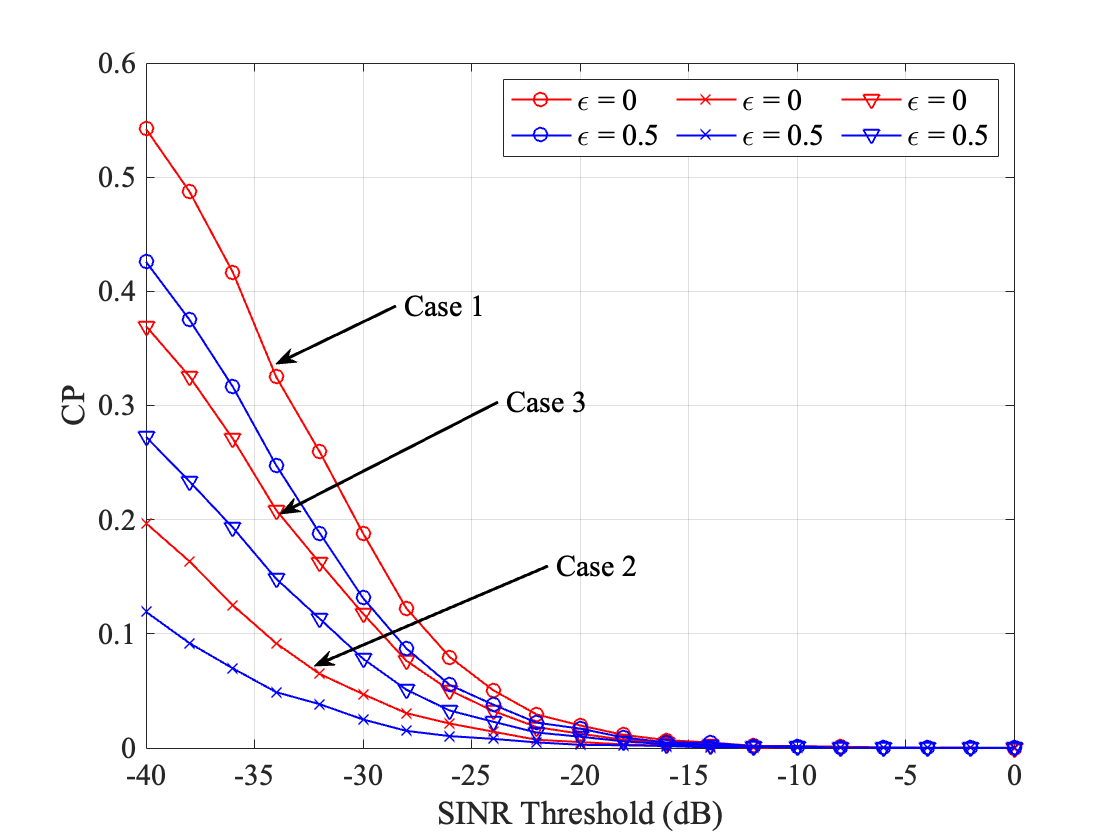}
\caption{The tendency of CP obtained by the Monte Carlo simulation vs. SINR threshold \& $\epsilon$.}
\label{fig:case123epsilonm}
\vspace{-0.3cm}
\end{figure}

\begin{figure}[!t]
\centering
\includegraphics[width=3in]{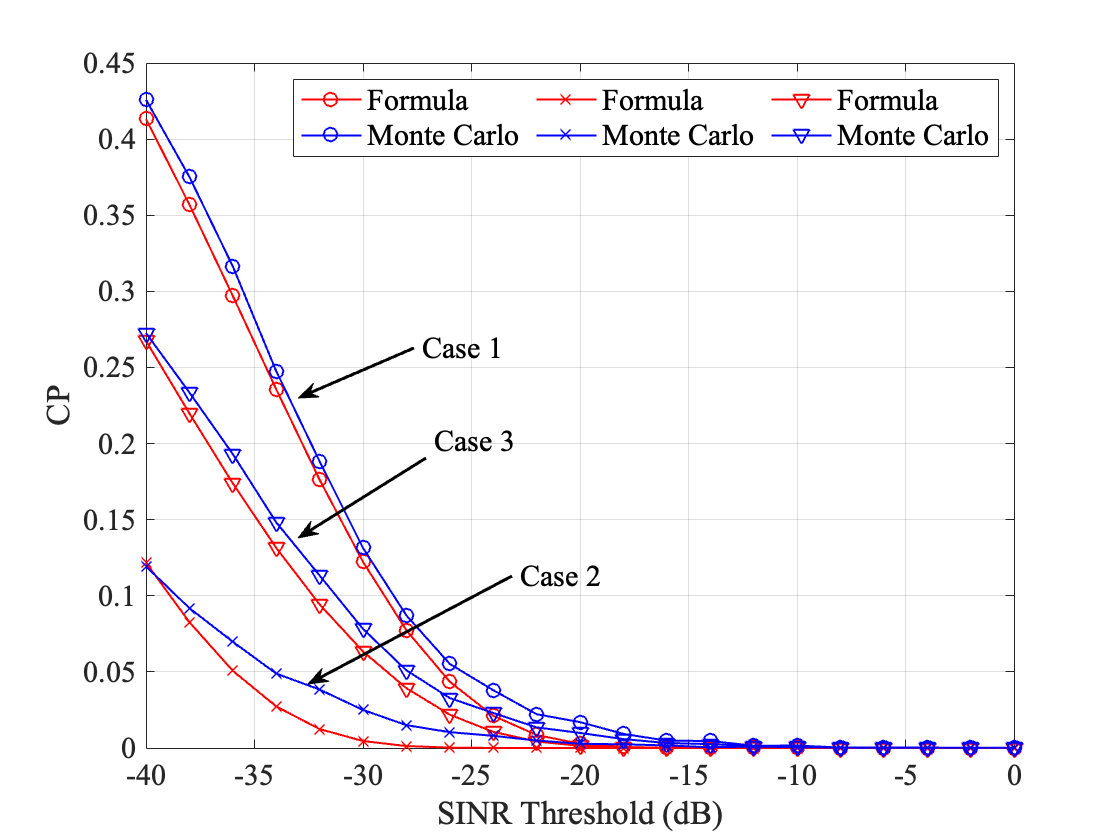}
\caption{A comparison of CP obtained by the theoretical derivation and the Monte Carlo simulation when varying SINR threshold in three cases ($\epsilon = 0.5$).}
\label{fig:case3epsilonc}
\vspace{-0.5cm}
\end{figure}

\begin{figure}[!t]
\centering
\includegraphics[width=3in]{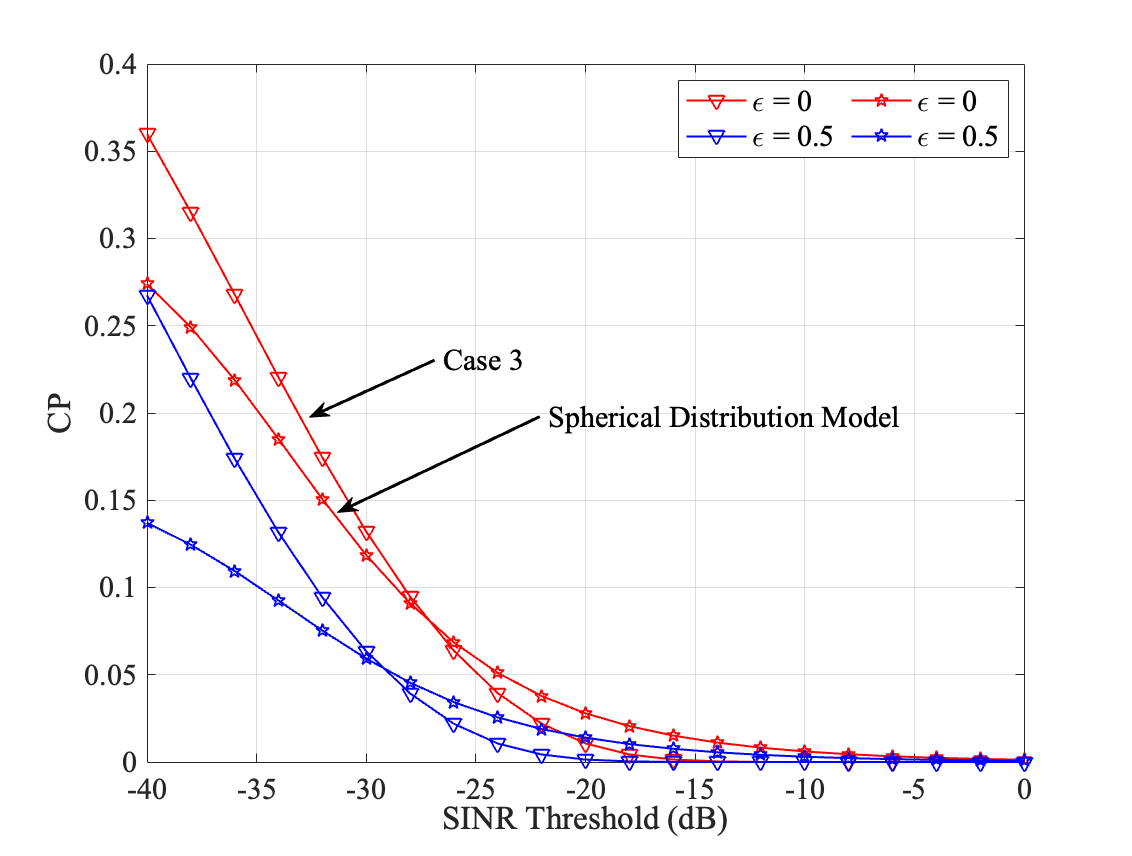}
\caption{A comparison of CP obtained by the theoretical derivation between Case 3 and spherical distribution model.}
\label{fig:case3spheref}
\vspace{-0.3cm}
\end{figure}

\begin{figure}[!t]
\centering
\includegraphics[width=3in]{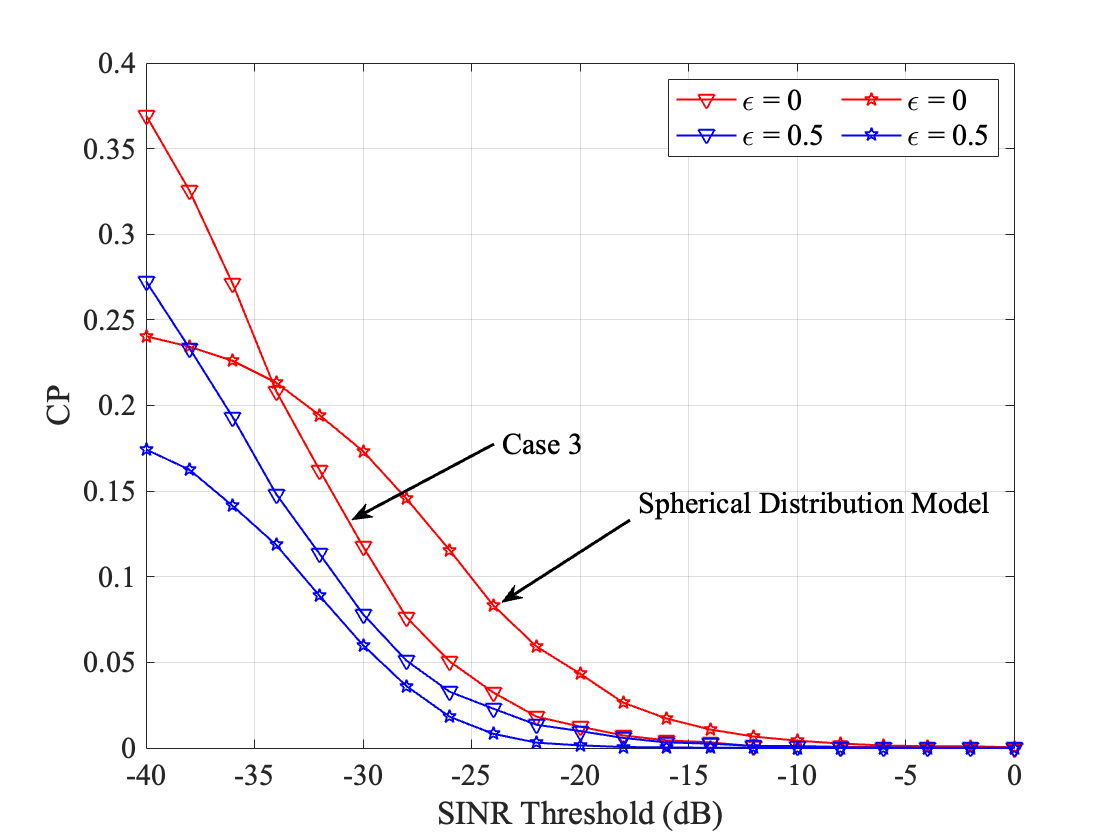}
\caption{A comparison of CP obtained by the Monte Carlo simulation between Case 3 and spherical distribution model when varying SINR threshold \& $\epsilon$.}
\label{fig:case3SphereM}
\vspace{-0.3cm}
\end{figure}

{\subsection{The Impact of Parameter $\epsilon$ on CP}}
{The power control factor is used to adjust the transmit power to adapt to different path loss conditions and maintain the signal strength of the communication.} 

Figs. \ref{fig:case123epsilonf} and \ref{fig:case123epsilonm} present the curves of CP with different $\epsilon$ versus $T_s$ in three cases when other parameters are the same. The shared observation from the cases is that, with an increasing $\epsilon$, the CP diminishes. As $\epsilon$ rises, the power compensation $\frac{H^{\alpha\epsilon + 1}_{TR}}{H_{TR} - H^{min}_{RX}}$ for receiving useful signals gradually grows. {However, this concurrently leads to an amplification in the compensation  $\sum_{z \in \mathcal{Z}}\frac{H^{\alpha\epsilon+1}_z}{H_z - H^{min}_{RX}}$ for multiple interference signals. Therefore, $\frac{H^{\alpha\epsilon + 1}_{TR}}{H_{TR} - H^{min}_{RX}}$ / $\sum_{z \in \mathcal{Z}}\frac{H^{\alpha\epsilon+1}_z}{H_z - H^{min}_{RX}}$ cannot explain the downward trend.} {To this end, we employ a dimensionality reduction scheme to analyze the impact of $\epsilon$ theoretically.} 

{After a series of simplifications to Eq. (\ref{cp_result}), we can obtain the only term $\frac{1}{1 + e^{y^{-(\alpha\epsilon+1)}}y^{\alpha\epsilon+1}}$ including $\epsilon$. The derivative is given by}
\begin{equation}
    \begin{split}
        \frac{\partial}{\partial \epsilon}\left(\frac{1}{1 + e^{y^{-(\alpha\epsilon+1)}}y^{\alpha\epsilon+1}}\right) &= -\frac{\frac{\partial}{\partial \epsilon}\left(e^{y^{-(\alpha\epsilon+1)}}y^{\alpha\epsilon+1}\right)}{\left({1 + e^{y^{-(\alpha\epsilon+1)}}y^{\alpha\epsilon+1}}\right)^2} \\
        &\hspace{-10em}= - \frac{-\alpha e^{y^{-(\alpha\epsilon+1)}}y^{-(\alpha\epsilon+1)}\ln{y} \cdot y^{\alpha\epsilon+1} + \alpha e^{y^{-(\alpha\epsilon+1)}} \cdot y^{\alpha\epsilon+1}\ln{y}}{\left(1 + e^{-x^{\alpha+1}}x^{-\alpha-1}\right)^2} \\
        &\hspace{-10em}= - \frac{\alpha e^{y^{-(\alpha\epsilon+1)}}\ln{y} \cdot (y^{\alpha\epsilon+1} - 1)}{\left(1 + e^{-x^{\alpha+1}}x^{-\alpha-1}\right)^2} < 0 \hspace{2em}(y^{\alpha\epsilon+1} > 1)
    \end{split}
\end{equation}

{Similarly, $-2\pi\lambda_n\int^{L^{max}_z}_{l_i}(1 - \int^{H^{max}_{RX}}_{H^{min}_{RX}}\frac{\mu}{\mu + sx^{-\alpha-1}\frac{y^{\alpha\epsilon + 1}}{y - H^{min}_{RX}}}\cdot f_{H_z}(y)dy)xdx$ shows the decreasing trend, which clarifies the result that a higher $\epsilon$ value corresponds to a decreased CP.}

Fig. \ref{fig:case3epsilonc} shows that the results of the theoretical derivation and the Monte Carlo method differ slightly, with differences in a reasonable and acceptable error range. This also indicates that the theoretical derivation can effectively predict the network performance. 
{Figs. \ref{fig:case3spheref} and \ref{fig:case3SphereM} depict the comparison of CP obtained by the theoretical derivation and the Monte Carlo simulation between Case 3 and spherical distribution model when varying SINR threshold. In this scenario, neither model is suitable for power compensation, although the results of Case 3 are slightly better.}

{\subsection{Summary of Results}}
In summary, it is inevitable to increase the number of radiosondes in order to collect more typhoon information, which means that the density of radiosondes in the same area will increase. However, it will introduce stronger interference. Thus, trade-offs should be made when considering increasing the number of radiosondes. In the typhoon environment, the path loss and rain attenuation are greatly related to the propagation distance; radiosondes that are farther away from the receiver suffer greater path loss, and the rain attenuation is more serious. By comparing the two cases, it can be seen that if the distribution of the radiosondes is more concentrated, it will cause stronger interference. When the distribution is more dispersed, the network performance will be improved. The power control factor has a negative effect on the CP of the radiosonde network because all the radiosondes are power-compensated. The farther propagation distance results in the greater compensation of power, and then the final CP is reduced.
\\

\section{Conclusion}
This paper theoretically analyzed the uplink connection performance of a radiosonde network deployed in a typhoon. 
Particularly, this paper separately modeled the distributions of horizontal and vertical distances from a radiosonde to its receiver. Given the horizontal and vertical distance distributions, the closed-form CDF and PDF expressions of a radiosonde's 3D propagation distance to its receiver were derived. The analytical expression of the uplink CP of a radiosonde was also obtained. Besides, this paper conducted extensive numerical simulations and comparisons between different models to validate the theoretical analysis. The influence of some parameters (e.g., power control factor and node density) were investigated.

{Nevertheless, there are several issues in the paper waiting to be addressed. The high-speed movement of radiosondes within a typhoon causes frequent Doppler shifts in the signal, which impacts the CP of the network. The calculation of network CP is limited to a specific moment or distribution, while the variable airflow within a typhoon may prevent the use of the same distribution or model over time. Radiosondes within a typhoon may have relatively short lifespans and are highly susceptible to damage.}

{To advance this research, we will focus on creating a model that integrates the lifecycle of radiosondes, using insights from fluid dynamics and the internal structure of typhoons to analyze their starting points and destinations. This approach enables modeling of radiosonde distributions across various regions within the typhoon. Further, by incorporating factors like Doppler shifts and additional attenuation effects into current channel models, it's possible to achieve network performance outcomes that better reflect actual conditions.}


 \appendix
 \subsection{CDF of the propagation distance $L_i$ in Case 1}
 With (\ref{eq17}), we first perform basic integration operations on $\mathcal{F}_{L_i}\left(l_i\right)$, it can be simplified as
 \begin{equation*}\label{eq28}
 \begin{split}
     \mathcal{F}_{L_i}\left(l_i\right) &= C_1\int^{R^{max}_{RX}}_0 2\pi\lambda_nue^{-\pi\lambda_nu^2}\int^{\sqrt{l^2_i-u^2}}_{H^{min}_{RX}} \frac{1}{2}e^{-\frac{v}{2}} dvdu \\
     &= C_1\int^{R^{max}_{RX}}_0 2\pi\lambda_nue^{-\pi\lambda_nu^2} \left(e^{-\frac{{H^{min}_{RX}}}{2}} - e^{-\frac{\sqrt{l^2_i-u^2}}{2}}\right) du \\
 \end{split}
 \end{equation*}
 \begin{equation}
     \begin{split}        
     &= C_1[e^{-\frac{{H^{min}_{RX}}}{2}}\int^{R^{max}_{RX}}_0 2\pi\lambda_nue^{-\pi\lambda_nu^2}du  -\\ &\hspace{4em}\int^{R^{max}_{RX}}_0 2\pi\lambda_nue^{-\pi\lambda_nu^2}e^{-\frac{\sqrt{l^2_i-u^2}}{2}}du]
 \end{split}
 \end{equation}

 At this point, both basic integral and improper integral appear in (\ref{eq28}). Due to the complicated process of dealing with the improper integral, we will calculate them separately. Denote the basic integral as $F_1$ and the improper integral as $F_2$.

 The closed-form expression of $F_1$ can take the following form
 \begin{equation}\label{F1}
 \begin{split}
     F_1 &= e^{-\frac{{H^{min}_{RX}}}{2}}\int^{R^{max}_{RX}}_0 2\pi\lambda_nue^{-\pi\lambda_nu^2}du \\
     &= e^{-\frac{{H^{min}_{RX}}}{2}}\int^{R^{max}_{RX}}_0 -e^{-\pi\lambda_nu^2}d\left({-\pi\lambda_nu^2}\right)\\
     &= -e^{-\frac{{H^{min}_{RX}}}{2}}e^{-\pi\lambda_nu^2}\left|^{R^{max}_{RX}}_0\right.\\
     &= e^{-\frac{{H^{min}_{RX}}}{2}}[1 - e^{-\pi\lambda_n{(R^{max}_{RX})}^2}]\\
 \end{split}
 \end{equation}

 For the improper integral, we eliminate the square root term in it via equal substitution. Let $\sqrt{l^2_i-u^2} = t$, $-2rdr = 2tdt$, $r \in [0,R^{max}_{RX}], t \in [\sqrt{l^2_i-(R^{max}_{RX})^2},l_i]$, we obtain
 \begin{equation}
 \begin{split}
     F_2 &= \int^{R^{max}_{RX}}_0 2\pi\lambda_nue^{-\pi\lambda_nu^2}e^{-\frac{\sqrt{l^2_i-u^2}}{2}}du\\
     &= \int^{l_i}_{\sqrt{l^2_i-(R^{max}_{RX})^2}} 2\pi\lambda_nte^{-\pi\lambda_n{(l^2_i-t^2)}}e^{-\frac{t}{2}}dt\\
     &= e^{-\pi\lambda_nl^2_i} \int^{l_i}_{\sqrt{l^2_i-(R^{max}_{RX})^2}} 2\pi\lambda_nte^{\pi\lambda_nt^2-\frac{t}{2}}dt\\
 \end{split}
 \end{equation}

 Subsequently, we transform the order of the exponential terms into the form of a complete sum of square. The terms unrelated to $t$ can be considered as constants and are moved outside the integral sign. The basic integral generated in this process can be directly calculated to obtain the result. The derivation is as follows
 \begin{equation}\label{F2}
 \begin{split}
     F_2&= e^{-\pi\lambda_nl^2_i} \int^{l_i}_{\sqrt{l^2_i-(R^{max}_{RX})^2}} \left[2\pi\lambda_n(t - \frac{1}{4\pi\lambda_n}) + \frac{1}{2}\right] \times \\
     &\hspace{6em} e^{\pi\lambda_n{(t - \frac{1}{4\pi\lambda_n})}^2 - \frac{1}{16\pi\lambda_n}}dt\\
     &= e^{-\pi\lambda_nl^2_i - \frac{1}{16\pi\lambda_n}}\int^{l_i}_{\sqrt{l^2_i-(R^{max}_{RX})^2}} [2\pi\lambda_n(t - \frac{1}{4\pi\lambda_n})\times \\ 
     &\hspace{6em} e^{\pi\lambda_n{(t - \frac{1}{4\pi\lambda_n})}^2} + \frac{1}{2}e^{\pi\lambda_n{(t - \frac{1}{4\pi\lambda_n})}^2}]dt\\
     &= e^{-\frac{l_i}{2}} - e^{-\pi\lambda_n(R^{max}_{RX})^2-\frac{\sqrt{l^2_i -(R^{max}_{RX})^2}}{2}} + \\
     &\hspace{6em}\int^{l_i}_{\sqrt{l^2_i-(R^{max}_{RX})^2}}\frac{1}{2}e^{\pi\lambda_n{(t - \frac{1}{4\pi\lambda_n})}^2}dt\\
 \end{split}
 \end{equation}

 Unfortunately, it is challenging to calculate the integral $\int^a_b e^{x^2}dx$ in a closed interval $[a,b]$. Taylor formula can be utilized to obtain its closed-form expression, denoted by $F_3$, with
 \begin{equation}\label{F3}
 \begin{split}
     F_3 &= \int^{l_i}_{\sqrt{l^2_i-(R^{max}_{RX})^2}}\frac{1}{2}e^{\pi\lambda_n{(t - \frac{1}{4\pi\lambda_n})}^2}dt\\
     &= \frac{e^{-\pi\lambda_n{l_i}^2}-\frac{1}{16\pi\lambda_n}}{2} \times \\
     &\hspace{-1em}\sum_{n=0}^{\infty}\frac{(\pi\lambda_n)^n}{n!}\frac{({l_i-\frac{1}{4\pi\lambda_n}})^{2n+1}-({\sqrt{l^2_i-(R^{max}_{RX})^2}-\frac{1}{4\pi\lambda_n}})^{2n+1}}{2n+1}\\
 \end{split}
 \end{equation}

 Ultimately, with (\ref{F1}), (\ref{F2}), and (\ref{F3}), the CDF of $L_i$ can be expressed as
 \begin{equation}
 \begin{split}
     \mathcal{F}_{L_i}\left(l_i\right) &= C_1[e^{-\frac{{H^{min}_{RX}}}{2}}[1 - e^{-\pi\lambda_n{(R^{max}_{RX})}^2}] - e^{-\frac{l_i}{2}}\\
     &+ e^{-\pi\lambda_n(R^{max}_{RX})^2-\frac{\sqrt{l^2_i -(R^{max}_{RX})^2}}{2}}- \\
     &\hspace{-4em}\sum_{n=0}^{\infty}\frac{(\pi\lambda_n)^n}{n!}\frac{({l_i-\frac{1}{4\pi\lambda_n}})^{2n+1}-({\sqrt{l^2_i-(R^{max}_{RX})^2}-\frac{1}{4\pi\lambda_n}})^{2n+1}}{2n+1}]
 \end{split}
 \end{equation}

 This completes the proof.

 \subsection{PDF of the propagation distance $L_i$ in Case 1}
 With (\ref{eq28}) and the theorem of differentiation under the integral sign, we obtain
 \begin{equation}
     \begin{split}
         f_{L_i}(l_i) &= \frac{dF_{L_i}(l_i)}{dl_i}\\
         &\hspace{-3em}= C_1\int^{R^{max}_{RX}}_0 \frac{\partial}{{\partial}l_i}[2\pi\lambda_nue^{-\pi\lambda_nu^2} \left(e^{-\frac{{H^{min}_{RX}}}{2}} - e^{-\frac{\sqrt{l^2_i-u^2}}{2}}\right)]du\\
         &\hspace{-3em}= C_1\int^{R^{max}_{RX}}_0 2\pi\lambda_nue^{-\pi\lambda_nu^2}e^{-\frac{\sqrt{l^2_i-u^2}}{2}}\frac{l_i}{2\sqrt{l^2_i-u^2}}du\\
     \end{split}
 \end{equation}

 Similarly, there are square root terms and exponential terms in the integral. Equal substitution ($\sqrt{l^2_i-u^2} = t$) can be also conducted here. Then we proceed with transforming the order to complete sum of a square and converting the improper integral term to Taylor formula sequentially. The process of integration and the final result are shown as below
 \begin{equation}
     \begin{split}
         f_{L_i}(l_i)&= C_1\int^{l_i}_{\sqrt{l^2_i-(R^{max}_{RX})^2}} 2\pi\lambda_nte^{-\pi\lambda_n{(l^2_i-t^2)}}e^{-\frac{t}{2}}\frac{l_i}{2t}dt\\
         &= C_1\pi\lambda_nl_ie^{-\pi\lambda_nl^2_i} \int^{l_i}_{\sqrt{l^2_i-(R^{max}_{RX})^2}}e^{\pi\lambda_nt^2-\frac{t}{2}}dt\\
         &= C_1\pi\lambda_nl_ie^{-\pi\lambda_nl^2_i-\frac{1}{16\pi\lambda_n}} \times\\
         &\hspace{-3em}\sum_{n=0}^{\infty}\frac{(\pi\lambda_n)^n}{n!}\frac{({l_i-\frac{1}{4\pi\lambda_n}})^{2n+1}-({\sqrt{l^2_i-(R^{max}_{RX})^2}-\frac{1}{4\pi\lambda_n}})^{2n+1}}{2n+1}
     \end{split}
 \end{equation}

 This completes the proof.

 \subsection{The Laplace transform of the interference in Case 1}
 With the definition of Laplace transform of interference  $\mathscr{L}_{\mathcal{I}^\prime_\mathcal{Z}}(s) = \mathbb{E}_{\mathcal{I}^\prime_\mathcal{Z}}[e^{-s{\mathcal{I}^\prime_\mathcal{Z}}}]$, we derive the expression of $\mathscr{L}_{\mathcal{I}^\prime_\mathcal{Z}}(s)$ as below
 \begin{equation}
 \begin{split}
     \mathscr{L}_{\mathcal{I}^\prime_\mathcal{Z}}(s) &= \mathbb{E}_{\mathcal{I}^\prime_\mathcal{Z}}[\exp(-\sum_{z \in \mathcal{Z}} s g_z A^{-1}_{r_z}L^{-\alpha}_{z}H^{\alpha\epsilon}_{z})]\\
     &\overset{(a)}{=} \mathbb{E}_{\mathcal{I}^\prime_\mathcal{Z}}[\exp(-\sum_{z \in \mathcal{Z}} s g_zL^{-\alpha-1}_{z}\frac{H^{\alpha\epsilon + 1}_{z}}{H_z - m})]\\ 
     &= \mathbb{E}_{g_z, L_z, H_z}[\prod_{z \in \mathcal{Z}}\exp(-\sum_{z \in \mathcal{Z}} s g_zL^{-\alpha-1}_{z}\frac{H^{\alpha\epsilon + 1}_{z}}{H_z - m})]\\
     &\hspace{-2em}\overset{(b)}{=} \mathbb{E}_{L_z, H_z}[\prod_{z \in \mathcal{Z}}\mathbb{E}_{g_z}[\exp(-\sum_{z \in \mathcal{Z}} s g_zL^{-\alpha-1}_{z}\frac{H^{\alpha\epsilon + 1}_{z}}{H_z - m})]]\\
     &\overset{(c)}{=}\mathbb{E}_{H_z}[\prod_{z \in \mathcal{Z}}\mathbb{E}_{L_z}[\frac{\mu}{\mu + sL^{-\alpha-1}_{z}\frac{H^{\alpha\epsilon + 1}_{z}}{H_z - m}}]]\\
     &\hspace{-4em}\overset{(d)}{=}\exp(-2\pi\lambda_n\int^{L^{max}_z}_{l_i}(1 - \mathbb{E}_{H_z}[\frac{\mu}{\mu + sx^{-\alpha-1}\frac{H^{\alpha\epsilon + 1}_{z}}{H_z - m}}])xdx)\\
 \end{split}
 \end{equation}
 where $(a)$ follows that $A_{r_z} = r\gamma_R\frac{(H_z - H^{min}_{RX})L_z}{H_z}$, $m = H^{min}_{RX}$, (b) follows from the independence of $g_z$, $(c)$ follows from the independence of $L_z$ and from the fact that $g_z \sim \exp(\mu)$, $(d)$ follows from the Probability Generating Functional (PGFL) of a HPPP. The mathematical expectation in the integral can then be written as
 \begin{equation}
 \begin{array}{l}
 {\mathbb{E}_{{H_z}}}\left[ {{\mu  \mathord{\left/
  {\vphantom {\mu  {\left( {\mu  + s{x^{ - \alpha  - 1}}\frac{{H_z^{\alpha  + 1}}}{{{H_z} - m}}} \right)}}} \right.
  \kern-\nulldelimiterspace} {\left( {\mu  + s{x^{ - \alpha  - 1}}\frac{{H_z^{\alpha  + 1}}}{{{H_z} - m}}} \right)}}} \right]\\
  = \int_{H_{RX}^{min}}^{H_{RX}^{max}} {\frac{\mu }{{\mu  + s{x^{ - \alpha  - 1}}\frac{{H_z^{\alpha  + 1}}}{{{H_z} - m}}}}}  \cdot {f_{{H_z}}}(y)dy
 \end{array}
 \end{equation}

 This completes the proof.

 \subsection{CDF of the propagation distance $L_i$ in Case 2}
 Similar to the derivations in Case 1, with (\ref{eq25}), we can compute the CDF of $L_i$ in Case 2 by
 \begin{equation}\label{eq38}
     \begin{split}
         \mathcal{F}_{L_i}\left(l_i\right) &= \frac{C_2}{(R^{max}_{RX})^2}\int^{(R^{max}_{RX})^2}_0 [e^{-(\frac{H^{min}_{RX}}{15})^6} - e^{-\frac{(l^2_i-t)^3}{15^6}}]dt \\
         &= \frac{C_2}{(R^{max}_{RX})^2}[(R^{max}_{RX})^2e^{-(\frac{H^{min}_{RX}}{15})^6}-\\
         &\hspace{10em}\int^{(R^{max}_{RX})^2}_0  e^{-\frac{(l^2_i-t)^3}{15^6}}dt]\\
     \end{split}
 \end{equation}
 where,
 \begin{equation}
     \begin{split}
         \int^{(R^{max}_{RX})^2}_0 e^{-\frac{(l^2_i-t)^3}{15^6}}dt &= \int^{(R^{max}_{RX})^2}_0 \sum^{\infty}_{n=0} \frac{1}{n!}{[-\frac{(l^2_i-t)^3}{15^6}]}^ndt \\
         &\hspace{-8em}= \sum^{\infty}_{n=0} \frac{(-1)^{n+1}}{n!15^{6n}\cdot(3n+1)}(l^2_i-t)^{3n+1}|^{(R^{max}_{RX})^2}_0\\
         &\hspace{-8em}= \sum^{\infty}_{n=0} \frac{(-1)^{n}}{n!15^{6n}\cdot(3n+1)}[(l^2_i)^{3n+1} - (l^2_i - (R^{max}_{RX})^2)^{3n+1}]
     \end{split}
 \end{equation}

 This completes the proof.

 \subsection{PDF of the propagation distance $L_i$ in Case 2}
 With (\ref{eq38}) and the theorem of differentiation under the integral sign, the PDF of $L_i$ in Case 2 is given by
 \begin{equation}
     \begin{split}
         f_{L_i}(l_i) &= \frac{d\mathcal{F}_{L_i}\left(l_i\right)}{dl_i}\\
         &= \frac{C_2}{(R^{max}_{RX})^2}\int^{(R^{max}_{RX})^2}_0 \frac{\partial}{\partial{l_i}}[e^{-(\frac{H^{min}_{RX}}{15})^6} - e^{-\frac{(l^2_i-t)^3}{15^6}}]dt\\
         &= \frac{C_2}{(R^{max}_{RX})^2}\int^{(R^{max}_{RX})^2}_0 [- e^{-\frac{(l^2_i-t)^3}{15^6}}\cdot(- \frac{3(l^2_i-t)^2}{15^6})\cdot2l_i]dt\\
         &= \frac{2C_2l_i}{(R^{max}_{RX})^2}\int^{(R^{max}_{RX})^2}_0 e^{-\frac{(l^2_i-t)^3}{15^6}}d(-\frac{(l^2_i-t)^3}{15^6})\\
         &= \frac{2C_2l_i}{(R^{max}_{RX})^2} (e^{-\frac{(l^2_i-(R^{max}_{RX})^2)^3}{15^6}} - e^{-(\frac{l_i}{15})^6})
     \end{split}
 \end{equation}

 This completes the proof.

\bibliography{referencess}

 





\end{document}